\documentclass[journal,10pt]{IEEEtran}
\usepackage{amssymb}
\usepackage{amsmath}
\usepackage{amsthm}
\usepackage{graphicx}
\usepackage{algorithm}
\usepackage{algorithmic}
\usepackage{booktabs}
\usepackage{cite}
\usepackage{bm}
\usepackage{float}
\usepackage{multirow}
\usepackage{color}
\usepackage{subfigure}
\usepackage{caption}
\usepackage{epstopdf}
\usepackage{hyperref}
\usepackage{fixltx2e}
\usepackage{longtable}
\usepackage{diagbox}
\usepackage{changepage}
\usepackage{comment}
\usepackage{cases}
\usepackage[short]{newoptidef}
\usepackage{makecell}
\usepackage{mathabx}

\usepackage[utf8]{inputenc}
\usepackage{array}
    \newcolumntype{P}[1]{>{\centering\arraybackslash}p{#1}}
    \newcolumntype{M}[1]{>{\centering\arraybackslash}m{#1}}

\usepackage{stackengine}

\usepackage{caption} 
\def\BibTeX{{\rm B\kern-.05em{\sc i\kern-.025em b}\kern-.08em
    T\kern-.1667em\lower.7ex\hbox{E}\kern-.125emX}}
\usepackage{balance}

\usepackage{enumitem}
\usepackage{cuted}
\usepackage{cryptocode} % for the protocol
\usepackage{fancybox,framed}

\usepackage{lineno,hyperref}
\modulolinenumbers[5]

\usepackage{multirow}
\usepackage{epsfig}
\usepackage{xcolor}
\usepackage[utf8]{inputenc}
\usepackage[short]{newoptidef}
\usepackage{makecell}
\usepackage{mathabx}
\usepackage{amsthm}
\theoremstyle{definition}

\newtheorem{mechanism}{Mechanism}

\theoremstyle{remark}

\theoremstyle{plain}
\newtheorem{theorem}{Theorem}
\newtheorem{lemma}{Lemma}

\usepackage{booktabs}

\usepackage{arydshln}

\begin{document}

% \title{Incentivising Peer-to-peer Multimedia Content Trading in IoV Networks}
\title{Decentralized Multimedia Data Sharing in IoV: A Learning-based Equilibrium of Supply and Demand}
\author{Jiani Fan, Minrui Xu, Jiale Guo, Lwin Khin Shar, Jiawen Kang, Dusit Niyato,~\IEEEmembership{Fellow,~IEEE}, and Kwok-Yan Lam,~\IEEEmembership{Senior Member,~IEEE}
 
\thanks{
Manuscript received April 6, 2023; revised Auguest 15, 2023; accepted 1 October 2023. Copyright (c) 2015 IEEE. Personal use of this material is permitted. However, permission to use this material for any other purposes must be obtained from the IEEE by sending a request to pubs-permissions@ieee.org.

Jiani Fan, Minrui Xu, Jiale Guo, Dusit Niyato, and Kwok-Yan Lam are with the School of Computer Science and Engineering, Nanyang Technological University, Singapore (e-mail: {jiani001, minrui001, jiale001}@e.ntu.edu.sg, {dniyato, kwokyan.lam}@ntu.edu.sg).

Lwin Khin Shar is with the School of Computing and Information Systems, Singapore Management University, Singapore (e-mail: lkshar@smu.edu.sg).

Jiawen Kang is with the School of Automation, Guangdong University of Technology, Guangzhou 510006, China (e-mail: kavinkang@gdut.edu.cn).
}
}

% The paper headers
\markboth{IEEE Transactions on Vehicular Technology,~Vol.~XX, No.~XX, XXX~2023}
{}

\maketitle

\begin{abstract}
% Internet of Vehicles (IoV) can revolutionize transportation systems by improving road safety, reducing traffic congestion, and enhancing the travel experience through ubiquitous cooperation among vehicles and infrastructure. 
% For cooperation in IoV, data sharing is a critical component, which is developed based on reliable and secure data-sharing mechanisms. The Internet of Vehicles (IoV) can revolutionize transportation systems through seamless collaboration between vehicles and infrastructure. Decentralized data sharing is a promising solution to ensure the robustness and effectiveness of IoV by improving security, privacy, reliability, and scalability. However, decentralized data sharing may not achieve the expected efficiency due to complete information observed by vehicles and infrastructure in the systems and thus introducing additional transmission latency. 
\textcolor{black}{The Internet of Vehicles (IoV) has great potential to transform transportation systems by enhancing road safety, reducing traffic congestion, and improving user experience through onboard infotainment applications. Decentralized data sharing can improve security, privacy, reliability, and facilitate infotainment data sharing in IoVs.} 
% However, decentralized data sharing may not achieve the expected efficiency due to incomplete information observed by vehicles and infrastructure, which can introduce additional transmission latency.
\textcolor{black}{However, decentralized data sharing may not achieve the expected efficiency if there are IoV users who only want to consume the shared data but are not willing to contribute their own data to the community, resulting in incomplete information observed by other vehicles and infrastructure, which can introduce additional transmission latency.}
% Therefore, in this paper, we {\color{black}propose} a decentralized data-sharing incentive mechanism based on multi-intelligent reinforcement learning to learn the supply-demand balance in markets and minimize transmission latency. 
\textcolor{black}{Therefore, in this paper, by modeling the data sharing ecosystem as a data trading market, we propose a decentralized data-sharing incentive mechanism based on multi-intelligent reinforcement learning to learn the supply-demand balance in markets and minimize transmission latency.}
Our proposed mechanism takes into account the dynamic nature of IoV markets, which can experience frequent fluctuations in supply and demand. We propose a time-sensitive \textcolor{black}{Key-Policy Attribute-Based Encryption (KP-ABE) mechanism} coupled with Named Data Networking (NDN) to protect data in IoVs, which adds a layer of security to our proposed solution. Additionally, we design a decentralized market for efficient data sharing in IoVs, where continuous double auctions are adopted. The proposed mechanism based on multi-agent deep reinforcement learning can learn the supply-demand equilibrium in markets, thus improving the efficiency and sustainability of markets. Theoretical analysis and experimental results show that our proposed learning-based incentive mechanism outperforms baselines {\color{black}by 10\%} in determining the equilibrium of supply and demand while reducing transmission latency {\color{black}by 20\%}. 
\end{abstract}


\section{Introduction}

% introduce IoV network, IoV, infotainment system
\textcolor{black}{The Internet of Vehicles (IoV) is a network of vehicles that enables the transmission of information between pedestrians, vehicles, and urban infrastructure \cite{nss,antsurvey,surveyIoV20} through collaboration between sensors, software, built-in hardware, and communication technologies.} The IoV network utilizes a variety of sensors, software, built-in hardware, and different types of connections to enable reliable and continuous communication\cite{nss}. With recent advancements in machine learning and automation, IoV systems are becoming increasingly sophisticated, offering capabilities such as real-time navigation guidance, automated driving, end-device connectivity, \textcolor{black}{and onboard infotainment services\cite{Karim2022}}. 

% introduce infotainment system
A significant accelerator for the commercialization of IoVs is the \textcolor{black}{emergence of smart infotainment systems\cite{9799964}}, such as the one in Tesla~\cite{tesla}. IoV infotainment systems play a critical role in this technology by delivering information and entertainment within vehicles via button panels, displays, and audio and video interfaces. \textcolor{black}{These systems offer a uniform user interface for both entertainment and driving assistance by connecting to onboard components through the Control Area Network (CAN)\cite{Kannadhasan_2021}}. IoV infotainment systems also serve as a platform for converting user input into messages that are transmitted over IoV networks via integrated Bluetooth, LTE, and Wi-Fi modules.
% challenges to be addressed
However, several challenges need to be addressed, such as the efficiency, security, and reliability of multimedia data sharing, which are essential to promote the full adoption of IoV systems.

% safety in multi-media should not be neglected
Firstly, data sharing in IoV requires reliable and secure mechanisms to ensure the privacy and confidentiality of data. Neglecting to secure multimedia data transfer in IoV networks may accidentally create a point of easy entry for social engineering attacks, \textcolor{black}{in which attackers can psychologically manipulate a target into acting on their behalf \cite{fi11040089, Al-Sabaawi2021}}. For instance, attackers may use vehicle-to-vehicle infotainment communication to disseminate false information about traffic conditions, deceive drivers into heavily populated areas of a highway, and interrupt traffic management by reporting false traffic information. These attacks can pose significant risks to the safety and security of IoV systems, making it imperative to develop secure and reliable data-sharing mechanisms. 

Meanwhile, decentralized data sharing emerges as a promising solution to address the security challenges, as it offers enhanced security and reliability compared to centralized data sharing. 
% decentralized data sharing still needs enhancement in efficiency
However, decentralized data sharing may not achieve the expected efficiency if there are participants who only want to consume the shared data but are not willing to contribute their data to the community, resulting in incomplete information observed by other vehicles and infrastructure, which can introduce additional transmission latency. For instance, decentralized data sharing frequently relies on voluntary resource contribution, where individuals are willing to contribute their network resources for the benefit of others. To achieve better resource utilization efficiency, incentive mechanisms can be used to \textcolor{black}{balance the demand and supply of network resources through a peer-to-peer data trading market \cite{9359135,Yogarayan_Razak_Azman_Abdullah_2021}}. In this market, users who demand resources pay a price for the resources they get, while users who have extra resources are compensated for supplying them. Perfect market efficiency and the best allocation of excess network resources are achieved when no resource is undervalued or overvalued.

\textcolor{black}{To determine the best pricing and resource allocation strategy,}
some auction mechanisms have been proposed, but they are not yet suitable for continuous and decentralized markets. For instance, second-price auctions can result in low efficiency and high budget costs due to a lack of coordination between buyers and sellers~\cite{deshmukh2002truthful}. In a typical second-price auction, a buyer may end up paying more than they expected, as the price they pay is equal to the second-highest bid. This can lead to a reduction in buyer participation and, consequently, a decrease in market efficiency. In addition, double auctions are susceptible to collusion between buyers and sellers, which can lead to lower efficiency and reduced welfare for other participants~\cite{anufriev2013efficiency}. Therefore, double auctions may not be suitable for continuous and decentralized markets, as they require a central auctioneer to coordinate the bidding, which can result in additional delays and costs. Finally, all of these auctions also primarily focus on the value of sharing data in markets, they cannot reduce transmission latency during data transmission~\cite{hui2022collaboration}.

% our contribution
\textcolor{black}{To address these challenges, in this paper, we model the data-sharing ecosystem as a data trading market and propose a decentralized data-sharing incentive mechanism based on multi-intelligent reinforcement learning to learn the supply-demand balance in markets and minimize transmission latency.} Furthermore, we propose a time-sensitive KP-ABE encryption mechanism coupled with NDN for protecting data in IoV to enhance the efficiency of data distribution and add an additional layer of security to our proposed incentive mechanism.

The main contributions of this paper are summarised as follows:
% our contribution
\begin{itemize}
\item We develop a Multi-Agent Deep Reinforcement Learning (MADRL)-based incentive mechanism to encourage decentralized data sharing amongst vehicles and RSUs. Our experiments have shown that the MADRL-based mechanism can converge to satisfactory performance in the decentralized continuous data-sharing market.
\item \textcolor{black}{We develop a decentralized continuous doubled-size market design that seamlessly integrates with our time-sensitive KP-ABE system, enabling the secure and efficient distribution of IoV infotainment data.}
%We propose a time-sensitive KP-ABE scheme that facilitates secure data sharing in IoVs. The proposed KP-ABE scheme prevents unauthorized access to shared infotainment data by controlling access based on attributes such as subscription status and validity time. %\textcolor{orange}{shar: to discuss this contribution}
%significantly strengthens the security of data sharing by limiting access to data based on validity time and reducing the risk of unauthorized access. 
\item We leverage Named Data Networking (NDN) for efficient circulation of data among the IoVs \textcolor{black}{and provide resource caching at the Road Side Units (RSUs)}, where data packets are self-contained and independent of the location at which they can be retrieved and transferred.

\end{itemize}

% structure of the paper + organization
The structure of this paper is in accordance with the following: 
% In Section \ref{sec: related}, we present the related work to introduce the context of the paper.
% In Section \ref{sec: overview}, we provide an overview of the system background on which we will base our approach.
% In Section \ref{sec: kp-abe}, we present the details of our time-sensitive KP-ABE scheme.
% In Section \ref{sec: decentralize}, we present the design of our decentralized continuous double-side market and introduce the roles of different actors in the market.
% In Section \ref{sec: incentive}, we present the multi-agent deep reinforcement learning-based incentive mechanism that aims to encourage user participation.
% In Section \ref{sec: experiment}, we present the experimental results of our scheme.
% In Section \ref{sec: conclusion}, we conclude our work in this paper. 
We present related works in Section \ref{sec: related}, an overview of the system background in Section \ref{sec: overview}, details on the proposed time-sensitive KP-ABE scheme in Section \ref{sec: kp-abe}, the decentralized continuous double-sided market design in Section \ref{sec: decentralize}, a multi-agent deep reinforcement learning-based incentive mechanism in Section \ref{sec: incentive}, experimental results of our scheme in Section \ref{sec: experiment}, and a conclusion in Section \ref{sec: conclusion}.

\section{Related Work} \label{sec: related}

\subsection{Internet-of-Vehicles}
% IoVs have a wide range of functions that they can carry out with seamless connectivity, such as providing onboard entertainment services and providing real-time navigational guidance. Given the extensive connectivity of IoV, it is believed to be an essential part for establishing a safe and efficient intelligent transportation system (ITS), providing firm ground for using autonomous vehicles in the future\cite{hv2x}. 
Among recent literature in IoV, reliable and scalable sharing of data has been a major concern \cite{reliableandscalable,routingalgo}. Various vehicular applications operate on shared information, such as cooperative awareness messages (CAMs) and basic safety messages (BSMs), to predict the behavior of other vehicles and optimize their decision-making \cite{hv2x}. However, building scalable and dependable communication networks in IoVs is difficult. The huge amount of data produced by vehicles and the dynamic nature of traffic flows make network elasticity a crucial factor to consider. This is exacerbated by the different latency tolerances that exist among the various types of messages that are communicated. For instance, real-time cooperative control and emergency information have strict time constraints, whereas infotainment applications can tolerate some latency\cite{reliableandscalable}.

\subsection{Caching With Named Data Networking}
\textcolor{black}{Infotainment data communication via mobile networks frequently relies on data caching techniques to achieve lower latency due to the high content volume of the data\cite{s22041387}}. 
% In highly dynamic mobile contexts like IoV systems, traditional communication protocols based on TCP/IP and IP will be less successful in fulfilling such efficiency objectives, especially as the number of devices in the network rises \cite{chen2020caching}. 
Due to the ongoing quest for better audio and video quality as well as the constrained network bandwidth, the IoV system frequently calls for a higher transmission rate and greater flexibility in data distribution techniques. NDN stands out as one of the suitable choices for resource optimization in these networks, since it offers compatibility with current routing protocols while also optimizing the use of communication resources \cite{zhang2014named}. It provides information-centric networking and is directly applicable to existing IP services like the Domain Name Service (DNS) and inter-domain routing policies. 
% IP routing protocols like BGP and OSPF may be easily modified to work with NDN. 
Instead of identifying data packets with source and destination addresses, NDN nodes recognize them by their names. 
% The data packets are independent of their retrieval and transfer locations and are self-contained. 
These features make it possible to cache material within the network for upcoming requests, \textcolor{black}{improving content mobility while doing away with the need for application-specific middleware\cite{9042284}}. 
% Additionally, NDN routers provide multi-path forwarding, which enables them to send a user request to multiple interfaces at once.
\textcolor{black}{While NDN applications aim to integrate data-centric security through the secure binding of names to ensure integrity guarantees \cite{Named_Data_Networking_(NDN)_2021}, an expansion of the NDN framework is essential to implement access controls for limiting the data usage boundaries within the confines of a single application's context.}

  \begin{figure*}[htb]
    \begin{center}
        \includegraphics[width=0.8\linewidth]{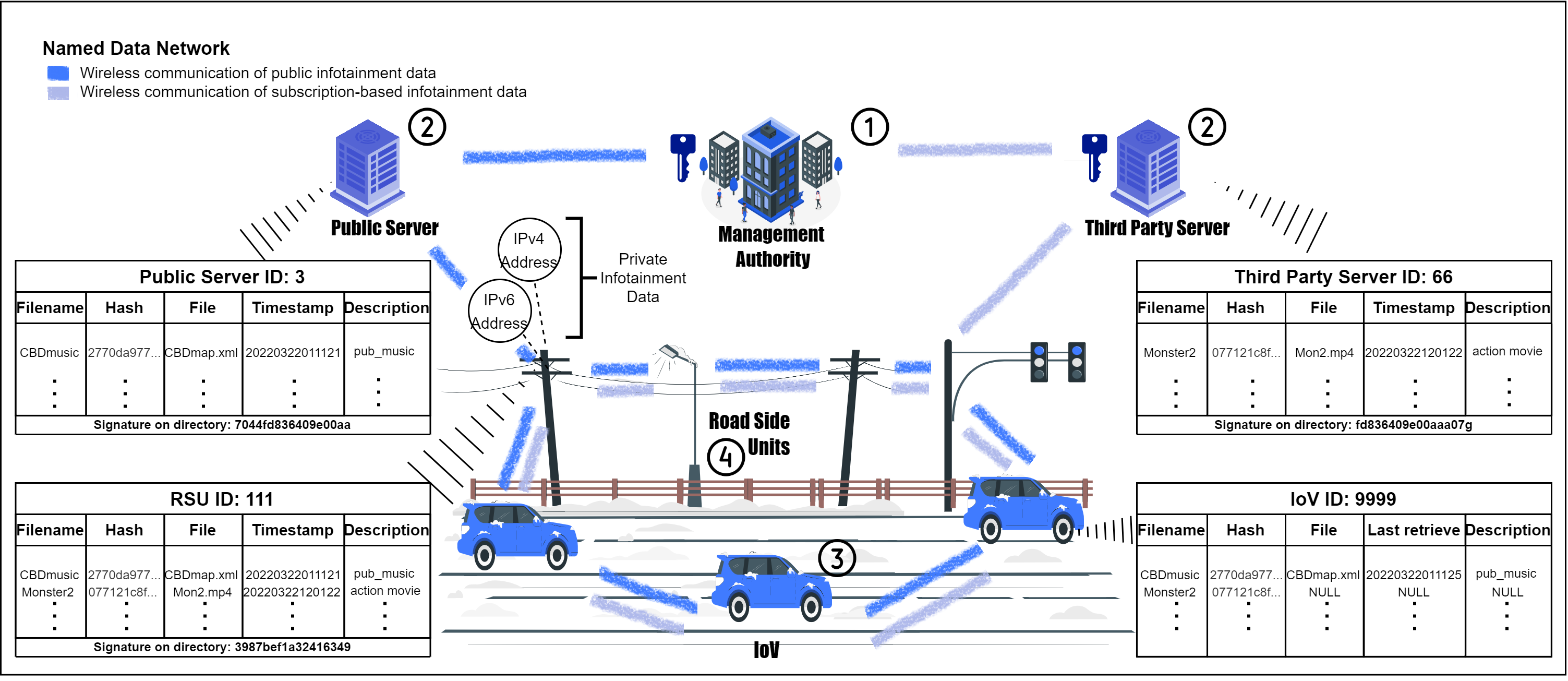}
        \caption{Illustration of the proposed NDN-based communication system\cite{nss}.}
        \label{fig:iov_context}
    \end{center}
    \vspace{-5mm}
    \end{figure*}
    
\subsection{Peer-to-peer Data Trading in IoV}
% A growing amount of data is generated and transferred among drivers and traffic infrastructures thanks to the IoV's extensive networking capabilities. This draws a variety of organizations to take part in the life cycle of data, mining the business possibilities in such data. 
IoV data trading frequently involves many participants, including data sellers, buyers, and providers, each of whom has their own set of interests \cite{blockchainp2p, 9211724}. 
% Due to its decentralized, immutable, secure, and transparent features, blockchain has been a well-liked candidate to help with data trading in such situations. However, efficiency issues are a common problem for blockchain-based data trading platforms\cite{debit}. 
% \cite{debit} is an illustration of an effective debt-credit mechanic to facilitate effective blockchain-based data trade on the Internet of Values.
\textcolor{black}{Some auction mechanisms have been proposed for dealing with peer-to-peer data trading in the IoV, but they are not effective in dealing with continuous and decentralized markets. Second-price auction mechanisms are where the highest bidder wins, but only pays the second-highest bid value plus one cent. It typically results in low efficiency and high budget costs due to a lack of coordination. While double auctions with multiple sellers and buyers require a central auctioneer and are not suitable for the decentralized market~\cite{deshmukh2002truthful, anufriev2013efficiency}. Furthermore, these auction mechanisms focus mainly on maximizing the market value and cannot reduce transmission latency, making them unsuitable for large multimedia sharing in IoV. To address these challenges, in this paper, we propose a decentralized data-sharing incentive mechanism based on multi-intelligent reinforcement learning to learn the supply-demand balance in the market and minimize transmission costs.
latency.}

\section{Overview of System Background} \label{sec: overview}

\textcolor{black}{On-the-go infotainment services are essential for the widespread adoption of IoV systems and the speed of information dissemination and resource sharing are crucial factors in the overall experience of IoV users. \textcolor{black}{With the large file size and stringent latency requirements for IoV communication\cite{9748067}}, peer-to-peer sharing and caching are desirable options to speed up the performance of infotainment services. However, we need strategies to avoid free-rider behaviors, where users act in their interests and are unwilling to help, and to achieve a market balance where the supply and demand of data are matched.} 

% In addition, the fundamental drivers behind smart transportation remain the need for a safer driving environment and an efficient transportation infrastructure. 
Due to the spontaneous and open environment in which these smart vehicles operate, secure and speedy communication with strong privacy protection and user anonymity also becomes a crucial requirement. \textcolor{black}{Meanwhile, most of the security proposals for IoV require a trusted central authority and focus on the authentication, authorization, and identity management}, which take less consideration for \emph{data-centric} security protection~\cite{nss, 9130160}. Furthermore, the conventional assumption of a trusted central authority may not always hold in the context of IoV, where there is an overwhelming number of potential entry points for attacks. An external smart device attached to any vehicle, a USB connected to a charging station, or even a wireless device that can access the same IoV network could be a launchpad for an attack on the central server. And this single point of failure could lead to system-wide malfunctioning.
\textcolor{black}{Hence, traditional centralized information distribution systems, which rely on a single central authority for secure resource allocation, are unsuitable for the extensively interconnected IoV networks. At the same time, while a balanced Content Distribution Network could prevent a single point of failure, it has difficulty adapting to a highly mobile environment such as the IoV network and the presence of multiple stakeholders who have differing interests.}

\subsection{Our proposal}
In response to the aforementioned findings, we propose a decentralized data-sharing incentive mechanism based on multi-intelligent reinforcement learning to learn the supply-demand balance in the market and minimize transmission latency. \textcolor{black}{By adopting a decentralized approach, we can overcome the limitation of a single point of failure in a centralized approach.} Furthermore, to enhance data security in decentralized data-sharing, we propose a time-sensitive KP-ABE for sharing multimedia data in IoV networks. With our proposed approach, IoVs can enjoy secure and efficient infotainment services under the familiar pay-as-you-use subscription model with peer-to-peer sharing of subscription materials.

A typical IoV transport system is composed of four basic parts, as shown in Fig. \ref{fig:iov_context}. The management authority (1) is a central organization that continuously analyzes traffic conditions and implements traffic control measures. It is also in charge of providing its subsidiary networks with critical traffic information, such as wide-area networks with traffic signaling and roadside infrastructure. Commonly, \textcolor{black}{centrally or remotely located public and third-party servers} are used to store the data created here (2). Due to the high mobility of IoVs, roadside units (4) and nearby vehicles (3) frequently serve as data relays to promote effective traffic information exchange and provide a smooth connection to IoVs.

% incentive mechanism
In our proposed system, shown in Fig. \ref{fig:iov_context}, \textcolor{black}{we employ an NDN system where files are identified by their names rather than IP addresses. File names can be retrieved from public directories in RSU or provided by content providers to user applications.}
\textcolor{black}{There are two types of servers: public servers that provide public infotainment content and third-party servers that additionally provide subscription-based infotainment content. These are some of the operations that servers perform:}

\begin{itemize}
    \item \textbf{Content directory:} All servers have an NDN directory of their infotainment content that includes the filename, hash value, file content, last modified timestamp, and a description of the file. The servers can share the whole directory or a subset of it with anyone on the resources that they are willing to share.
    \item \textbf{Content encryption:} For subscription services, third-party servers will encrypt subscription-based files using randomly generated 256-bit AES keys. Each AES key is encrypted using $\mathsf{Encrypt}(\mathrm{PK}, \mathcal{M}, \mathbb{T}_c, S)$ function in Section \ref{subsec: kp-abe} to produce a ciphertext $CT$. Afterwards, each $CT$ is appended to the corresponding encrypted file to produce a new file that is ready for circulation in the network. The hash value of the new file is re-calculated and stored as a new entry in the directory. This process is illustrated in Fig. \ref{fig:aes}.
    \begin{figure}[t]
        \begin{center}
            \includegraphics[width=0.65\linewidth]{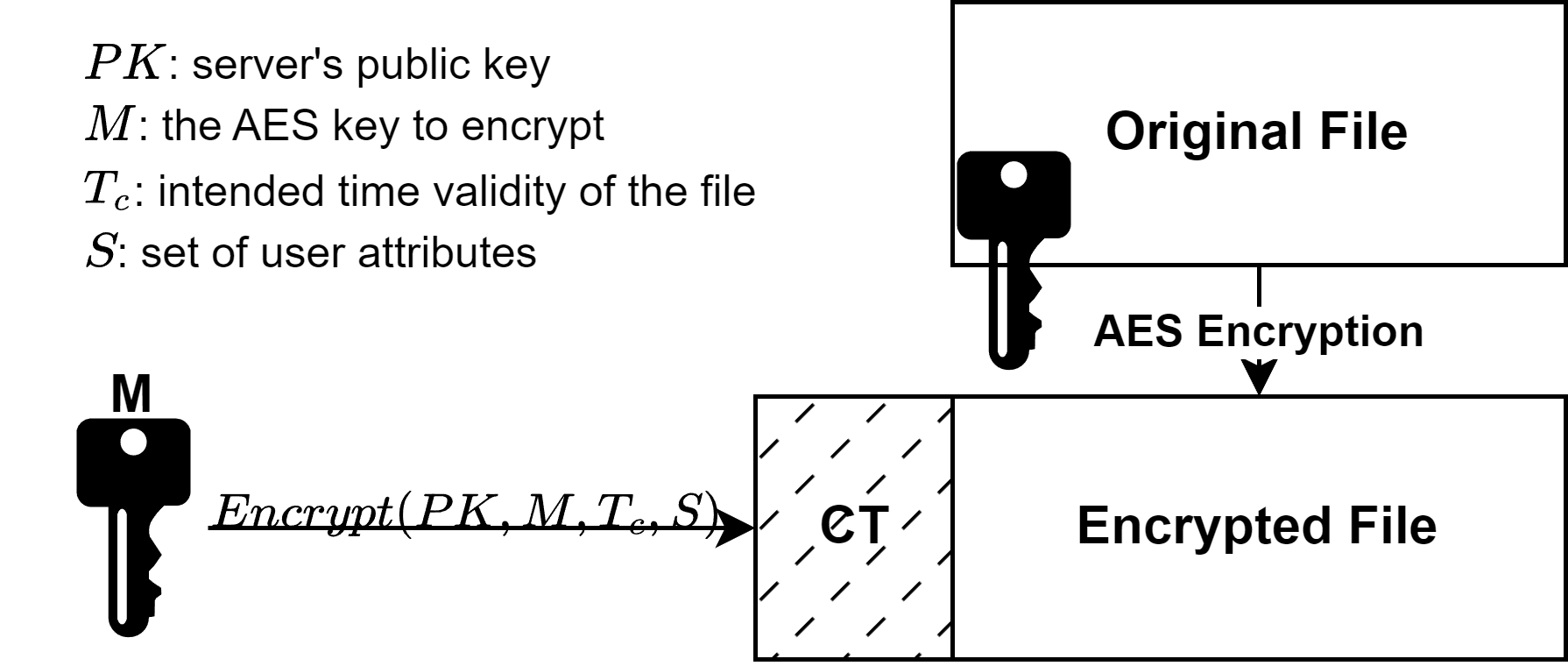}
            \caption{Content encryption at service provider's servers.}
            \label{fig:aes}
        \end{center}
        \vspace{-5mm}
        \end{figure}
        
    \end{itemize}
    \begin{itemize}
    
    \item \textbf{Resource distribution:} Directories of public infotainment content and encrypted subscription-based content are shared with RSUs. 
    \begin{itemize}
        \item A directory of popular subscription-based content is downloaded onto the user application installed on the vehicle when the user subscribes to the third-party service.
        \item The servers also fulfill any file requests from the RSUs. Note that only the encrypted version of subscription-based content will circulate on the network. Anyone in the network can own a copy of the encrypted file, but only valid subscribers can decrypt it.
        \item To speed up content distribution and reduce transmission delay, NDN-based content popularity and router level (CPRL) in-network caching strategy~\cite{8076542} is used to cache popular content on routers installed on RSUs.
    \end{itemize}
\end{itemize}

\textcolor{black}{RSUs are stationary infrastructures that act as an intermediary between servers and vehicles. They have an important role to play in facilitating data sharing between local vehicles. These are some of the operations that RSUs perform:}
\begin{itemize}
    \item \textbf{Storage and caching}: A directory containing resource names and file hashes is kept in RSU, digitally signed by the RSU to guarantee its integrity. Additionally, RSUs can store popular public infotainment and traffic data to improve user experience~\cite{khelifi2019named}. NDN-based CPRL in-network caching is also employed to cache the content according to its popularity level on the RSU's router to reduce content retrieval latency.
    \item \textbf{Content retrieval}: When the RSU receives a vehicle data request, it checks if it has a local copy of the requested resource. If it does not, it requests the resource from the relevant servers and then relays it to the vehicles. 
    \item \textbf{Auctioneer}: In the global decentralized IoV data-sharing market, each RSU holds a local submarket where buyers and sellers are vehicles under its coverage. The bids of multiple buyers and sellers can be aggregated in the RSU to determine allocation and pricing rules.
\end{itemize}

\textcolor{black}{IoVs are highly mobile vehicles that are traveling at different speeds and with different bandwidth conditions from time to time. There could be a significant amount of excess bandwidth when little communication or computation occurs within the vehicle, or there could be a high utilization of resources when it performs bandwidth-consuming operations. These are some of the operations that IoVs perform:}
\begin{itemize}
    \item \textbf{Directory update}: Vehicles can request an update for the directory from the RSUs and check the RSU's digital signature against the public infrastructures' pre-loaded public key certificates\cite{ELKHALIL2021101885}.
    \item \textbf{Content retrieval}: Vehicles are permitted to request files from nearby vehicles or any RSU that holds the files, whichever is closest to them. In this system, every vehicle can request or provide data at any time. If the communication cost for content retrieval from the RSU is higher than the cost of performing P2P data sharing, the vehicle is incentivized to participate in the local data-sharing market via a double auction.
    \item \textbf{P2P market participation}: To participate in the market, the buyers and sellers submit their buying bids and selling bids to their local RSUs, respectively. 
    \begin{itemize}
        \item Each vehicle that provides data aims to earn as much as possible without hindering its own performance while fulfilling other peers' requests.
        \item On the other hand, each vehicle that requests data seeks to receive the requested file as quickly as possible, at the lowest possible price, or a combination of both.
    \end{itemize}
    \item \textbf{Content decryption}: When the data sharing of subscription-based content is complete, the receiver will perform the $\mathsf{Decrypt}(CT, \mathrm{SK}_{(\mathrm{ID},A,T)})$ function in Section \ref{subsec: kp-abe} on the ciphertext $CT$ to obtain the AES key to the content. The AES key is then used to decrypt encrypted content to gain access.
    
\end{itemize}

% %%% caching
% %%% how auction is used with KP-ABE

% directory updates

In the following subsection, we first classify data exchange in an IoV network and illustrate the security and efficiency requirements of different data types, which leads to key design considerations for security schemes that have not yet been fully realized in many generic security schemes. Then, we introduce the gossip protocol that we used to collect market information for the construction of our incentive mechanism.

% classification of data exchange 
\subsection{Data Classification}

In general, there are six major types of data that are communicated in the IoV network when we classify data communications according to their security requirements.

\begin{enumerate}
\item \textbf{Vehicle-to-everything (V2X) private information exchange:} This includes private or sensitive information about the vehicle or the user that is not intended to be made available to the general public. Since these communications can reveal the user's personal information, it is crucial to preserve data integrity and secrecy. 

\item \textbf{Traffic control messages:} These are instantaneous traffic control messages, such as those for traffic lights and emergencies. These messages must meet strict standards such as \textcolor{black}{ISO 22737\:2021, IEEE 1609.0, SAE J2945/1\_202004} for message integrity and availability and are time-sensitive. Any attempt to obstruct them could have fatal repercussions \textcolor{black}{such as car accidents or delays in emergency handling that result in casualties.} 

\item \textbf{Public traffic data:} This includes geographical maps, alerts about infrastructure maintenance, traffic congestion status information, and even the locations of petrol stations, fire stations, and auto repair shops. There is no need for confidentiality protection since this information is meant for all drivers. On the other hand, for road users to make informed decisions, the accuracy and accessibility of this information are crucial.

\item \textbf{Publicly accessible infotainment data:} All publicly accessible infotainment data is included in this category, such as public websites, social media platforms, and publicly available content from service providers.

\item \textbf{subscription-based infotainment data:} This category includes subscription-based third-party infotainment content serviced by external service providers. Due to the subscription nature of these services, user authentication and subscription status verification are vital to avoid free riders—\textcolor{black}{unpaid users who enjoy content intended for subscribers only}. In addition, the quality of these paid services, such as low latency and availability, is important, which necessitates effective caching schemes.

\item \textbf{Private infotainment data:} This includes personal or restricted infotainment data that is not for public viewing. This will therefore require rigorous security measures to guarantee its confidentiality and integrity. 

\end{enumerate}

% explain the table, different types of data have different QoS
% In Table \ref{tab: dataCategoryRequirement}, we look at the various QoS of each type of data transmission in terms of confidentiality, integrity, long-term availability, and short-term availability for different data categories in the IoV. The IoV systems' constrained computing resources motivate us to prioritize the security needs in accordance with the kind of data delivered and structure differentiated security protocols to fulfill different Quality-of-Service (QoS).

To maintain confidentiality and integrity for private data, communication data such as ``private infotainment data" and ``V2X private information exchange" should be safeguarded through strict authentication measures~\cite{9406119}. While integrity and accessibility are given top priority for ``traffic control messages" with minimal latency. On the other hand, public data types, such as public traffic and infotainment data, should be available to all. Thus, public data have less criticality for availability. The same methods that protect the aforementioned data kinds are inappropriate for these public data types. Similarly, subscription-based infotainment data require additional access restrictions based on subscription status to protect copyrights.
\textcolor{black}{In this KP-ABE scheme, we focus on the protection of subscription-based infotainment data sharing in the NDN-based IoV communication network. While public and private infotainment data should be secured, NDN networking provides an inherent capability to guarantee data integrity, which fulfills the security requirement for public infotainment data. At the same time, private data requires stringent security protection with less focus on data sharing, such as authentication, access control and encryption to confine access to one or a few individuals. Hence, our scheme aims to improve overall resource utilization efficiency by providing NDN-based access control for subscription-based infotainment data. By restricting access to only the subscribers, we can allow anyone in the network to own the same subscription-based resource and perform peer-to-peer trading with subscribers in need, thereby reducing content retrieval latencies without concern about the problem of free riders.}
% For easy dissemination of public data, edge nodes are often employed for caching such information. Also, peer-to-peer sharing schemes can be employed to speed up the circulation of public information. Similarly, subscription-based infotainment data is often cached in edge nodes for user experience, requiring additional access restriction by subscription status.

\subsection{Peer-to-Peer (P2P) Gossip Algorithm}

% \begin{table*}[t]
% \centering
% \begin{tabular}{ p{2cm}p{2cm} p{2cm} p{2cm}p{4cm}}
% \toprule
% \textbf{ID} & \textbf{Node Type} & \textbf{Price Per Unit} & \textbf{Waiting Time} & \textbf{Publish Time}  \\
% \midrule
% F9283738OW1 & IoV & 2 Credits/KB & - & 2022-02-02 12:12:32\\
% G2832738GH2 & IoV & 2.5 Credits/KB & - & 2022-02-02 12:12:32\\
% X32398B & RSU & - & 60 s & 2022-02-02 12:12:32\\
% \bottomrule
% \end{tabular}
% \caption{A sample table entry for Gossip protocol.}
% \label{tab: gossip}
% \end{table*}

%latest transaction price and nearby request
Gossip algorithms are distributed algorithms that spread messages around in a network, where the receiver subsequently spreads the messages to its neighboring nodes, thereby propagating the messages in the network\cite{1238221}. By allowing messages to propagate through the network, we eliminate the need for a central authority to provide timely information about the local market, thereby reducing the complexity of the network and avoiding any single-point failure\cite{8637812}. Hence, IoVs in the network can obtain the latest transaction price in the local market, as well as any recent request for which they can bid, without any private communication between peers.
% There are three basic types of gossip algorithms:
% \begin{enumerate}
% \item random gossip protocol, a node randomly selects a number of directly neighboring nodes to send the gossip message.
% \item geography-based gossip protocol, where a node can choose to spread the gossip message to neighbors beyond one hop distance.
% \item broadcast gossip protocol, a node sends the gossip message to all direct neighbors.
% \end{enumerate}

During each communication round, nodes are allowed to modify or append new information locally and spread the newly modified messages, depending on the protocol design. In this paper, we adopt a random gossip protocol~\cite{9539193} to circulate the latest transaction price and any new file request from neighbors. Each IoV stores and updates a local transaction price table that documents the recent transactions that occur near it and the current queuing time at the neighboring RSU (estimated waiting time for any new incoming request to be fulfilled). The same copy of the transaction price table is maintained throughout the network using gossip algorithms. 
% A sample of the gossip entry table is found in Table \ref{tab: gossip} which also includes entries from RSU.

\begin{figure*}[hbt]
    \begin{center}
    \includegraphics[width=1\linewidth]{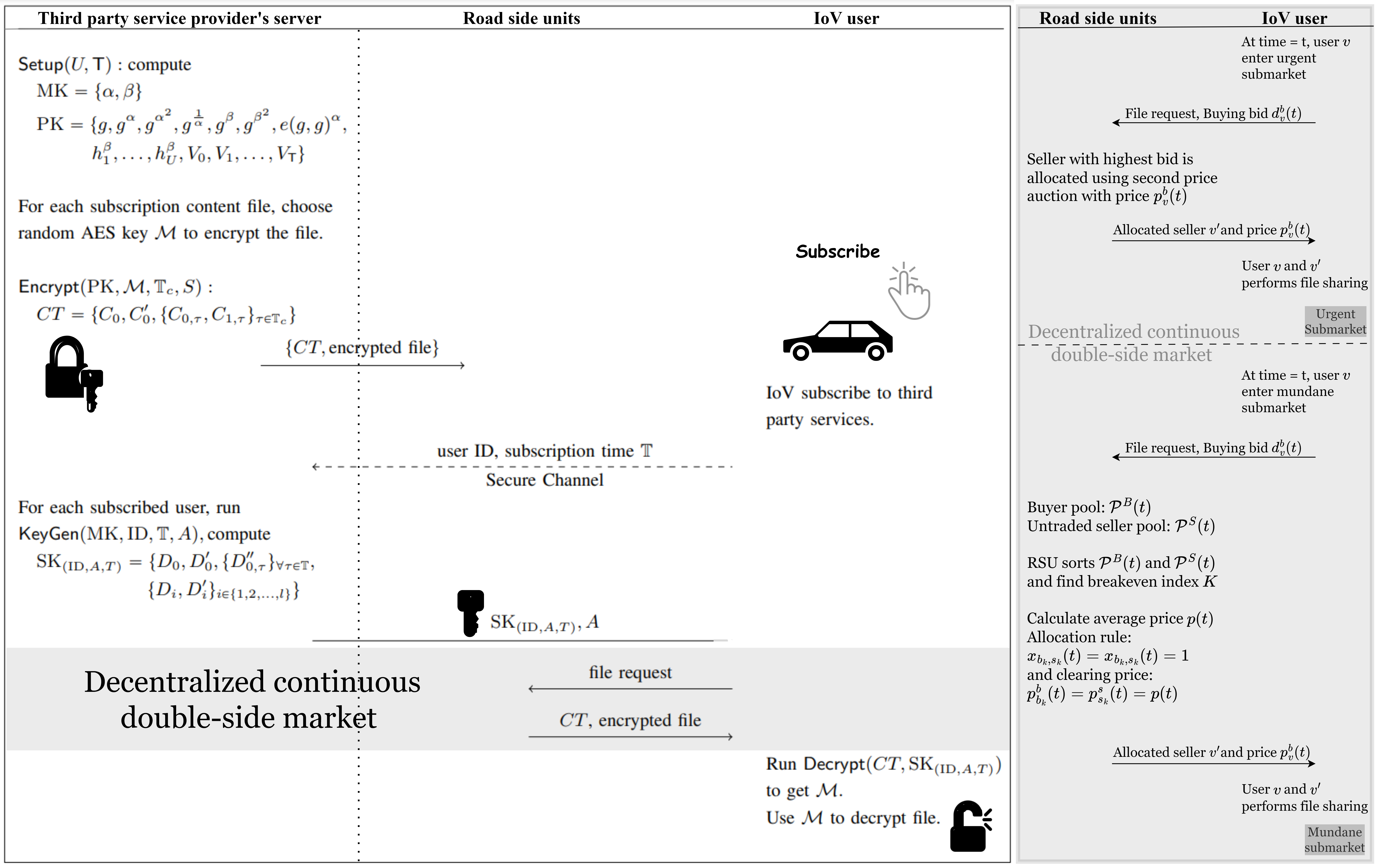}
    \caption{\textcolor{black}{Workflow of the time-sensitive KP-ABE scheme integrated with our proposed mechanism.}}
    \label{fig:exchanges}
     \end{center}
     \vspace{-5mm}
\end{figure*}

\section{Time-sensitive KP-ABE} \label{sec: kp-abe}
Security and efficiency are the primary concerns when we are transforming our vehicle transportation systems into IoV-based intelligent systems. In this section, we will introduce the time-sensitive KP-ABE proposed in previous work\cite{nss}, which aims to enhance the security of data sharing by limiting access to data based on validity time and user subscription attributes and reducing the risk of unauthorized access. We have considered the following factors when developing our time-sensitive KP-ABE for secure data sharing in the IoVs:

\paragraph{\textbf{Differentiated security}} Critical traffic information, such as traffic control and emergency messages, requires comprehensive cryptographic protections, while public data has a lower demand for confidentiality protection. Meanwhile, subscription-based data is semi-public and require access restrictions, where a group of subscribers can share a common set of media content. Thus, we can improve the system's resource efficiency in providing security protection by tailoring mechanisms according to the security requirements of data types. In our system, private or sensitive communication is protected using traditional authentication schemes\cite{ali2019authentication, 9552176}, public data is hashed, and the public data directory containing the filename and hash values is signed using the public key certificates of the traffic management authority for integrity protection. In addition, we employ a time-sensitive encryption scheme to avoid free riders for subscription-based content.

\paragraph{\textbf{Minimizing IoV network latency}} Traffic data is often time-sensitive and spontaneous, which puts enormous pressure on traffic systems to communicate information in a timely manner. However, because infotainment data files are often in megabytes or gigabytes, transmitting large files with a very small amount (e.g., 30–40 ms) of network latency will be a big challenge for the reliability of the network. Hence, we need a strategic content distribution system that could facilitate the easy distribution of large content files. Thus, we employ NDN and cache both subscription-based and public data at the RSUs to reduce the hop distance between users and the data.

\paragraph{\textbf{Re-usability of encrypted files}} To minimize network delay, it is better to retrieve infotainment data files from nearby devices. While all users of the road have access to information that is publicly available, any encrypted subscription content should be decryptable by all users whose policies comply with the decryption requirements. As a result, subscription content should follow the general guidelines that apply to privileged users. As a result, we propose a time-sensitive KP-ABE scheme that enables the decryption of files for all users with matching subscription attributes and whose subscription validity period fully overlaps the necessary time validity on files. This allows the exact duplicate to be sent throughout the network and decrypted upon request by any privileged user.

% \subsection{Proposed Approach for Infotainment Data Exchange}

% \textbf{Public infotainment data:} A directory of public infotainment data is available at RSUs, which includes their filename, hash value, update-timestamp, and content description. IoV users can always retrieve the latest copy of the directory from any RSU, verify the RSU's digital signature on it, and download whatever content the IoV wants using the filename. Due to the high velocity of the IoVs, large files are broken into small parts. IoVs could request the missing parts from neighbouring IoVs or RSUs, whichever is closer.

% \textbf{subscription-based infotainment data:} 

% \textbf{Private infotainment data:} Private infotainment data can be highly sensitive. Hence, stringent authentication mechanisms have to be in place to prevent unauthorized access to such data at the cost of significant overheads, which is not the focus of this paper.

\subsection{Details of the KP-ABE scheme} \label{subsec: kp-abe}
The focus of our KP-ABE scheme is to facilitate secure and efficient communication of subscription-based infotainment data, which is often substantial in size. 
% content protection
In summary, to protect content copyrights, service providers can encrypt their content using our KP-ABE scheme, where files are encrypted with user attributes that represent their subscription details. \textcolor{black}{With this scheme, users can only decrypt files whose validity time falls within their subscription time. When the validity of a user's access policy expires, the user loses the ability to decrypt the files until they resubscribe from the service provider. Coupled with our NDN architecture, any node in the network can own the file, but only those with a valid subscription can decrypt it. Therefore, our KP-ABE scheme could provide time-based access control and eliminate the problem of periodic user revocation due to expired subscriptions. Furthermore, our scheme provides a strong foundation for peer-to-peer data sharing, since any IoV on the network can own the encrypted resource and participate in data sharing.}

The stored third-party content directories \textcolor{black}{and cached resources} in the RSUs can be updated regularly by service providers through authorized communication. The encrypted file can be sent across untrusted channels, and its integrity can still be verified since both the filename and the hash value of the encrypted file have been disclosed to the RSUs via directories. When a user subscribes, third-party service providers can preload the most popular content \textcolor{black}{onto the user's application agent installed on the vehicle}, and an on-board directory of popular content can be deployed. The application could send a request to the RSU to help obtain the user's request from the server if it is not already on the list. The application could ask any nearby vehicle or RSU for the file once they have the filename and hash value of the required content.

\textcolor{black}{A detailed workflow on our proposed data exchange can be found in Fig. \ref{fig:exchanges} where the KP-ABE scheme \cite{nss} is integrated with our proposed mechanism.} Inspired by the work \cite{liu2018time, liu2020new}, we use a Hierarchical Identity-based Encryption (HIBE) technique to control the time validity of the infotainment files. In general, the time periods are represented by a hierarchical tree, which has one topmost root node and at most three-level non-root nodes. Each node in the first level of the tree represents a year, and its child in the second level represents a month in this year. The third-level nodes represent days.
% Therefore, each node of the tree can represent a time period, i.e. a day, a month, or a year. 

Note that we also adopt the set-cover approach to select the minimum number of nodes to represent the valid time periods. The use of HIBE with the set-cover approach can effectively reduce the number of key generations required to represent each time period. In particular, a user should obtain the corresponding attribute in their access policies for those time periods. The user can only decrypt files whose validity time period falls within his/her subscription time period, meaning files that have a validity time period that is equal to or is a subset of their subscription time. When the validity of a user's access policy expires, the user loses the ability to decrypt until they resubscribe from the service provider.
\begin{itemize}
    \item \textcolor{black}{User A purchases a subscription service from 2022-JUL-01 to 2022-SEP-02.}
    \item \textcolor{black}{User A's valid time periods under the set-cover is \{2022-JUL, 2022-AUG, 2022-SEP-01, 2022-SEP-02\}.}
    \item \textcolor{black}{User A can decrypt files that have an expiration date between 2022-JUL-01 and 2022-SEP-02.}
    \item \textcolor{black}{User A cannot decrypt files after 2022-SEP-02.}
\end{itemize}

%For example, if a user purchases a subscription service from 2022-JUL-01 to 2022-SEP-02, the tree of his valid time periods should contain four nodes, including {2022-JUL, 2022-AUG, 2022-SEP-01, 2022-SEP-02}. Therefore, this user should obtain the corresponding attribute in their access policies for those four nodes by using HIBE in order to decrypt the ciphertext within the valid time periods. For example, the user can only decrypt files whose validity time period falls within his/her subscription time period, meaning files that have a validity time period that is equal to or is a subset of 2022-JUL-01 to 2022-SEP-02. When the validity of a user's access policy expires (after 2022-SEP-02), the user loses the ability to decrypt until they resubscribe from the service provider.
\begin{table}[t]
\centering
\caption{\textcolor{black}{A summary of math notation and symbols.}}
\renewcommand{\arraystretch}{1.1}
\begin{tabular}{p{2.5cm}<{\centering}|l}
\hline
Symbol & Description \\ \hline
$U$             & the number of attributes \\ \hline
$\mathsf{T}$    & the depth of the time tree \\ \hline  
$\mathrm{MK}$          & the master key \\ \hline
$\mathrm{PK}$          & public parameters \\ \hline
$\mathbb{G}_1$ & a bilinear group of prime order $p$  \\ \hline
$g$              & a generator of $\mathbb{G}_1$  \\ \hline
$h_1, h_2, \ldots,h_U$   & random elements chosen from $\mathbb{G}_1$ \\ \hline
$V_1, V_2,\dots,V_\mathsf{T}$   & random elements chosen from $\mathbb{G}_1$ \\ \hline
$\alpha$, $\beta$      & random numbers chosen from $\mathbb{Z}_p$ \\ \hline
$\mathrm{ID}$ & a user's pseudo-identity  \\ \hline
$\mathbb{T}$  & a set-cover of a user's decryptable time periods \\ \hline
$\tau$              & a z-ary representation of a time element \\ \hline
$A$           & a LSSS access structure  \\ \hline
$M$              & an $l\times n$ matrix  \\ \hline
$\rho$              & a mapping function \\ \hline
$\boldsymbol{v}$ & a random masking vector in $\mathcal{Z}_p^n$ \\ \hline
$\omega$              & an encryption exponent  \\ \hline
$\lambda_i (i=1,2,\dots, l)$ & the shares of $\omega$ \\ \hline
$\mathrm{SK}$ & a private key of a user \\ \hline
$\mathcal{M}$ & a plaintext message \\ \hline
$\mathbb{T}_c$   & a set of decryptable time periods of a message \\ \hline
$S$              & a set of attributes of the message  \\ \hline
$CT$          & a ciphertext \\ \hline

\end{tabular}
\label{tab:my-table}
\vspace{-5mm}
\end{table}

Our scheme consists of a 4-tuple of algorithms, denoted $(\mathsf{Setup}, \mathsf{KeyGen}, \mathsf{Encrypt}, \mathsf{Decrypt})$, of which the construction details are shown below. A summary of math notation and symbols is provided in Table \ref{tab:my-table}: 
\begin{itemize}[leftmargin=10pt]
\item $\mathsf{Setup}(U,\mathsf{T})$: It takes the number of attributes $U$ and the depth of the time tree $\mathsf{T}$ as input, outputs the public parameters $\mathrm{PK}$ and a master key $\mathrm{MK}$. In specific, given the depth $\mathsf{T}$, each time period is represented as a $z$-ary string $\{1,z\}^{\mathsf{T}-1}$, i.e, \{2022, 09, 22\}.
The algorithm chooses a bilinear group $\mathbb{G}_1$ of prime order $p$ with a random generator $g$, and randomly selects $U$ elements from the group, i.e. $h_1, h_2, \dots, h_U \in \mathbb{G}_1$. Besides, it also randomly chooses $\alpha,\beta \in \mathbb{Z}_p$ and $V_0, V_1, \dots, V_{\mathsf{T}} \in\mathbb{G}_1$. 
Then, it outputs
\begin{equation*}
\begin{aligned}
&\mathrm{MK}=\{\alpha, \beta\},\\
\mathrm{PK}=&\{g, g^{\alpha}, g^{{\alpha}^2}, g^{\frac{1}{\alpha}}, g^{\beta}, g^{\beta^2}, e(g,g)^{\alpha}, h_1^{\beta},\dots,h_U^{\beta}, V_0,\\ &V_1, \dots, V_{\mathsf{T}} \}.
\end{aligned}
\end{equation*}

\item $\mathsf{KeyGen}(\mathrm{MK}, \mathrm{ID}, \mathbb{T}, A)$: For a user with a pseudo-identity $\mathrm{ID}$ (A different pseudo-identity is generated by the service provider based on the user's identity for every purchase) and a set-cover of decryptable time periods, denoted as $\mathbb{T}$ that each of the elements in $\mathbb{T}$ can be represented as a $z$-ary representation $\tau=\{\tau_1, \tau_2, \dots, \tau_k\}\in\{1,z\}^k$ where $k<\mathsf{T}$, give the master key $\mathrm{MK}=\{\alpha, \beta\}$ and the LSSS access structure $A=\{M, \rho\}$, where $M$ is an $l\times n$ matrix and $\rho$ is a mapping function that maps each row of $M$ into an attribute. This algorithm outputs a private key $\mathrm{SK}_{(\mathrm{ID}, A, T)}$ for this user according to the following operations. At first, the algorithm chooses a random masking vector $\boldsymbol{v}=\{w, y_2,\dots,y_n\}\in \mathbb{Z}_p^n$ to share the encryption exponent $w$. Besides, it computes $\lambda_i = \boldsymbol{v} \cdot M_i$ for $\forall i\in \{1, 2, \dots, l\}$, i.e. $M_i$ is the $i$-th row vector of $M$. Here $\{\lambda_i\}$ are the shares of the secret $w$ according to $M$. Then this algorithm can calculate
\begin{equation}
\label{equ:sk}
\begin{aligned}
&D_0=e(g,g)^{\alpha w},~ D_0'=g^{\frac{w}{\alpha}},~\\ &\left\{D_{0,\tau}''=\left(V_0\prod_{j=1}^kV_j^{\tau_j}\right)^w\right\}_{\forall \tau\in \mathbb{T}},~\\
&D_i=g^{\beta\lambda_i},~
D_i'=\left(gh_{\rho(i)}^{\beta}\right)^{\lambda_1\mathrm{ID}}
\end{aligned}
\end{equation}
% $D_0'=g^{\frac{w}{\alpha}}$, $D_0''=\left(V_0\prod_{j=1}^kV_j^{\tau_j}\right)^w$, $D_i=g^{\beta\lambda_i}$, $D_i'=\left(gh_{\rho(i)}^{\beta}\right)^{\lambda_1\mathrm{ID}}$, 
and produce the private key \textcolor{black}{$SK$ for a user with pseudo-identity $ID$, access structure $A$ and a time validity period of $T$. The user will receive this private key upon subscription, and it is stored on user applications that are installed on the vehicle. This private key can be used to decrypt files that are encrypted using $\mathsf{Encrypt}(\mathrm{PK}, \mathcal{M}, \mathbb{T}_c, S)$.}
$$\mathrm{SK}_{(\mathrm{ID}, A, T)}=\{D_0, D_0', \{D_{0,\tau}''\}_{\forall \tau\in \mathbb{T}}, \{D_i, D_i'\}_{i\in\{1,2,\dots,l\}}\}$$

\item $\mathsf{Encrypt}(\mathrm{PK}, \mathcal{M}, \mathbb{T}_c, S)$: This is the algorithm that uses the public key $\mathrm{PK}$ generated by the $\mathsf{Setup}$ algorithm to encrypt a plaintext message $\mathcal{M}$\footnote{$\mathcal{M}$ is generally the 256-bits AES key used to encrypt the actual content because the size of the actual content is generally larger than the maximum size of the message that can be encrypted by ABE schemes.} associated with a set of attributes $S$ and a set of decryptable time periods $\mathbb{T}_c$. 
The set $S$ consists of attributes such as movie rating and subscription tier (e.g., platinum, gold, and silver). 
The set $\mathbb{T}_c$ consists of some time elements $\tau=\{\tau_1', \tau_2', \dots, \tau_{k_{\tau}}'\}\in\{1,z\}^{k_{\tau}}$ where $k_{\tau}<\mathsf{T}$. The set $\mathbb{T}_c$ is determined by the service provider. For example, if the provider decides that the content is valid for a particular period, $\mathbb{T}_c$ will cover that period so that only users who subscribed for this period will be able to decrypt.

The algorithm chooses a random $x \in \mathbb{Z}_p$ and for $\forall \tau\in \mathbb{T}_c$, it chooses a random $v_{\tau}\in \mathbb{Z}_p$. It then computes 
\begin{equation}
\label{equ:ct}
\begin{aligned}
&C_0=\mathcal{M}\cdot e(g,g)^{\alpha x},~ 
C_0'=g^{{\alpha}^2x},~
C_{0,\tau}=g^{v_{\tau}},\\
&C_{1,\tau} = g^{\alpha x}g^{{\beta}^2}\left(V_0\prod_{j=1}^{k_y}V_j^{\tau_j'}\right)^{v_{\tau}}
\end{aligned}
\end{equation}
where $k_y=(g^{\beta} h_y^{\beta})^{-1}$ for $y\in S$. Finally, it outputs the ciphertext $CT=\{C_0,C_0',\{C_{0,\tau}, C_{1,\tau}\}_{ \tau\in\mathbb{T}_c}\}$ along with the time periods $\mathbb{T}_c$. 
%The set of attributes $S$ consists of the attributes of the privileged users to whom the content provider intends to deliver the file. For example, if a movie is only available to ``platinum" members, the AES key to the file is encrypted with a ``platinum" user attribute.
\textcolor{black}{The ciphertext $CT$ contains the AES key to the encrypted resource files, and only users with validity within $\mathbb{T}_c$ are able to decrypt the $CT$ and obtain the AES key to the file.}

\item $\mathsf{Decrypt}(CT, \mathrm{SK}_{(\mathrm{ID},A,T)})$: This algorithm takes as input the ciphertext $CT$ and a user's private key $\mathrm{SK}_{(\mathrm{ID},A,T)}$, and outputs $\bot$ if any one of the following situations occurs:
\begin{enumerate}
    \item $S$ does not satisfy the access structure $A=\{M,\rho\}$.
    \item $T$ is not completely covered in $\mathbb{T}_c$, i.e. $\tau_T$ and all its prefixes are not in $\mathbb{T}_c$.
\end{enumerate}
Otherwise, let $I=\{i:\rho(i)\in S\} \subset \{1,2,\dots,l\}$, there exists a set of constants $\{\omega_i\in\mathbb{Z}_p\}_{i\in I}$ satisfying that $\sum_{i\in I}\omega_i \lambda_i=w$, where $\lambda_i$ are valid shares of a secret $w$ according to $M$. Finally, this algorithm can decrypt $CT$ as
$$\frac{C_0\cdot e(D_0'',C_{0,\tau}\cdot e(C_0',D_0'))}{e(C_0',g^{1/\alpha})\cdot\prod_{i\in I}\left(e\left(C_{1,\tau},(D_i')^{\frac{\omega_i}{\mathrm{ID}}}\right)\cdot e(D_i,k_{\rho(i)})^{\omega_i}\right)}.$$ \textcolor{black}{By performing this step, the user can obtain the AES key to the encrypted resource file and thereby decrypt the resource file and access the content.}
\end{itemize}

With our time-sensitive KP-ABE scheme, we provide the necessary security protection to perform decentralized data sharing in the IoVs.
\textcolor{black}{While our KP-ABE and the adoption of NDN serve as the framework for peer-to-peer information exchange, sending large media files would require a lot of bandwidth and time, which might prevent the sender from using such bandwidths for other purposes during the transmission period. In spite of the fact that such a peer-to-peer system increases the efficiency of the market, it is challenging to persuade individuals to adopt it without incentives. With our incentive mechanism, our architecture can be more practical and effective by encouraging users to participate in the peer-to-peer data-sharing system.}
In the following sections, we will present the market design and incentive mechanisms for encouraging user participation in decentralized data sharing.

\section{Decentralized Continuous Double-side Market Design} \label{sec: decentralize}

% % model considerations
The incentive mechanism aims to encourage the decentralized sharing of infotainment data among vehicles and RSUs in the proposed network, which is illustrated in Fig.~\ref{fig:iov_context}. In the NDN network, large infotainment files can be divided into multiple smaller chunks tagged using their filename and hash, and vehicles can request only those chunks that they do not have. Vehicles can request files from RSUs via V2I communication or from other vehicles that already have the files via V2V communication. However, due to the dynamic traffic environment, it is hard to design mechanisms that both maximize the total utility of vehicles and minimize transmission latency during data sharing. Therefore, in this paper, we first design a continuous decentralized double-sided market that is hierarchically composed of the urgent submarket and mundane submarket for secure data sharing in IoVs. Then, we propose an intelligent mechanism based on multi-agent deep reinforcement learning that achieves the balance of supply and demand in the market by learning entry strategies for buying vehicles. 

We consider a continuous market model, where the time in the system $\mathcal{T}=\{1, \ldots, t, \ldots, T\}$ is divided into $T$ time slots. In this market model, learning agents can choose to enter or exit the market to obtain maximum utility, and this interaction is structured as a non-cooperative game that emphasizes learning. In this game, each participant (vehicle) aims to maximize their credit earnings by utilizing excess bandwidth onboard. The credit mentioned here indicates the credit earned for every kilobyte of data sent to the requester, either by relaying data on behalf of the peer or by directly transferring data to the requester. 

There are three types of participants in our data sharing scheme: a) data servers (including public and third-party infotainment service providers); b) roadside units (RSUs) and c) vehicles. These participants are introduced in detail below.

% There are three main parties that participate in this data-sharing process: data servers (both public and third-party infotainment data suppliers), roadside units (RSUs), and vehicles.

\subsection{Data Servers}
{\color{black}In this system, the cloud data server provides mainly two types of infotainment content, i.e., public infotainment content and subscription infotainment content. Public infotainment content is available through cloud data servers and delivered to vehicles via RSUs. Subscription infotainment content is provided by third-party servers and is optionally cached at the RSUs. Third-party service providers may utilize the storage of RSUs to cache popular contents and reduce the transmission delay of their content to their subscribers.}

\subsection{Roadside Units (RSUs)}
RSUs are resourceful and can handle all data requests from vehicles. In this system, we consider a set of $N$ RSUs $\mathcal{N}$. {\color{black}As illustrated in Fig.~\ref{fig:mechanism}, in the global decentralized IoV data-sharing market, each RSU holds a local submarket where buyers and sellers are vehicles under its coverage.}  When the RSU receives a data request, it checks whether it has a local copy of the requested resource. If it does not, it requests the resource from the relevant servers and then relays it to the vehicles. This indicates that there may be a significant {\color{black}waiting delay} when many requests are made to the RSUs. Therefore, it is necessary to inform vehicles of the expected wait time to fulfill a request at the RSUs so that they can consider this.
When vehicles request the data from RSUs, there is a queuing latency $l_n$ for RSU $n \in \mathcal{N}$.
% The RSUs are always the sellers in their local markets, whose valuation $u_n$ is calculated according to the waiting latency $l_n$.
For each transmission of infotainment content, transmission latency is also considered part of the cost of the communication. \textcolor{black}{The vehicles are encouraged to perform peer-to-peer data sharing when the communication cost for data retrieval is lower than that of the RSU. RSUs also play an important role as the auctioneer for their local content-sharing market, where the bids of multiple buyers and sellers can be aggregated at the RSU to determine the allocation and pricing rules.}

\subsection{Vehicles}
In this system, we consider a set of $V$ vehicles $\mathcal{V}=\{1,\ldots, V\}$ where vehicles travel at different speeds on the road under the coverage of different RSUs. At each time slot $t$, if vehicle $v$ participates in one of the local markets of the nearest RSU $n$ then $g_{v,n}(t)=1$, and otherwise $g_{v,n}(t)=0$. The vehicles can request data or share data with other vehicles within the local market. 

At time slot $t$, the set of vehicles in the system is divided into the set of buyers and the set of sellers according to their preference in trading, i.e., $\mathcal{V} = \mathcal{V}^B(t) \bigcup \mathcal{V}^S(t)$, where $\mathcal{V}^B(t)$ represents the set of vehicle buyers and $\mathcal{V}^S(t)$ represents the set of vehicle sellers. The winning vehicle buyers in the $\mathcal{V}^B(t)$ receive the content from the corresponding vehicle sellers. Meanwhile, the losing vehicle buyers request the content directly from the RSUs regardless of waiting and transmission delay.  The transmission rate of direct communication between RSU $n\in \mathcal{N}$ and vehicle $v \in \mathcal{V}$ is $R_{n,v}^{d}(t)$.
\begin{itemize}
\item Each vehicle has an NDN directory that consists of the file names and hash values of the files, \textcolor{black}{i.e. the tables shown in Fig. \ref{fig:iov_context}}. Entries of public information in the directory are obtained from public servers, while entries of subscription-based resources are obtained from third-party service providers.
\item Vehicles can check for local resources or send communication requests for infotainment content when the user on board initiates the request.
\item Vehicles can have different bandwidths and may be occupied with a different amount of {\color{black}communication bandwidth at different time slots}. Hence, the willingness of an IoV peer to relay files for its peers differs according to the excess bandwidth of the vehicle from time to time. 
\end{itemize}

In this system, every vehicle can either request or provide data. Each vehicle that provides data aims to earn as much as possible without hindering its performance while fulfilling other peers' requests. On the other hand, each vehicle that requests data seeks to receive the requested file as quickly as possible, at the lowest possible price, or a combination of both. Vehicles that belong to the set of sellers can also act as relays to send data to vehicles in the set of buyers. Thus, we can obtain the cooperative transmission rate between vehicle $v \in \mathcal{V}^B(t)$ and relay vehicle $\bar{v} \in \mathcal{V}^S(t)$ as $R^{c}_{v,\bar{v}}(t)$. At a given time slot $t$, vehicle $v$ has a valuation $u_v(t)$ for the opportunity of content sharing. This valuation is affected by the size of the content $c_v(t)$ when the vehicle is in the set of buyers or by transmit power $p_v(t)$ when the vehicle is in the set of sellers. In addition, the valuation of buyers follows the economic law of ``diminishing marginal returns," while the valuation of sellers follows the economic law of ``increasing marginal costs.” {\color{black} Therefore, the utility function of requesting vehicle $v\in\mathcal{V}^B(t)$ can be represented as
\begin{equation}
    \mu_v(t)(b_v(t)) = u_v(t) - p_v^b(t),
\end{equation}
where $b_v(t)$ is the buying bid of vehicle $v$ and $p_v^b(t)$ is the buying charge of vehicle $v$. In addition, the utility function of selling vehicle $v\in\mathcal{V}^S(t)$ can be represented as
\begin{equation}
    \mu_v(t)(b_v(t)) = p_v^s(t) - u_v(t),
\end{equation}
where $b_v(t)$ is the selling bid of vehicle $v$ and $p_v^b(t)$ is the selling revenue of vehicle $v$.}

In some cases, the relay vehicle can help deliver the requested information. For instance, if a requester is too far away from a vehicle that has the requested information, the vehicle may ask another vehicle that is closer to the requester to deliver the information. In this case, the price offered by the requester for the vehicle should include the cost of the intermediary vehicle for delivering the information. In the next section, we will introduce our approach using this market design.

\section{Multi-agent Deep Reinforcement Learning-based Incentive Mechanism} \label{sec: incentive}

In this system, we consider a decentralized market, i.e., there is no centralized auction but each RSU acts as an auctioneer for its local content-sharing market. The system described here is used for decentralized content-sharing markets, where each RSU acts as an auctioneer for its local market. {\color{black}The reason is that the market is double-sided and the bids of multiple buyers and sellers can be aggregated at the RSU to determine the allocation and pricing rules.} {\color{black} The MADRL approach involves traders acting as intelligent agents that learn participating strategies by interacting with local markets. This allows for flexibility in selecting strategies for each individual market based on its current status~\cite{liu2022welfare}.}

In this system, the system detects a user's request for content, which then sends a request to the outside. The requested content may be available from some users in the local market, who have {\color{black}extra bandwidth} and are willing to participate in the resource transfer. To participate in the market, the buyers and sellers submit their buying bids and selling bids to their local RSUs, respectively. Before calculating the allocation and prices, the auctioneers of the RSUs need to collect information on the local market using gossip to provide the supply and demand matrix.
The allocation rules $\Pi$ consist of the supply matrix $X(t) \subseteq \{0,1\}^{|\mathcal{V}^B(t)|\times \mathcal{V}^S(t)|} = \{\mathbf{x}_1(t),\ldots,\mathbf{x}_{|\mathcal{V}^B(t)|}(t)\}$ and demand matrix $Y(t)\subseteq \{0,1\}^{|\mathcal{V}^B(t)| \times | \mathcal{V}^S(t)|} = \{\mathbf{y}_1(t),\ldots,\mathbf{y}_{|\mathcal{V}^S(t)|}(t)\}$. The supply vector of selling vehicle $v'$ can be represented as $\mathbf{y}_{v'} =\{y_{v',1}(t), \ldots, y_{v',|\mathcal{V}^S(t)|}(t)\}$, while the demand vector of buying vehicle $v$ can be represented as $\mathbf{y}_{v} =\{y_{v,1}(t), \ldots, y_{v,|\mathcal{V}^B(t)|}(t)\}$. With this information, the RSUs can calculate the allocation and prices based on the auction mechanism $\mathcal{M}=(\Pi, \Psi)$, where $\Pi=(X,Y)$ is the allocation rules and $\Psi=(\mathbf{p}^b, \mathbf{p}^s)$ is the pricing rules. For buyers, a truthful bid $b_v^t$ corresponds to their true valuation $u_v$, meaning that $b_v(t) = u_v$. The auction mechanism ensures that their payment $p^b_v(t)$ is less than or equal to their true valuation $u_v$, i.e., $p^b_v(t) \leq u_v$. For sellers, a truthful ask $a_v(t)$ represents their true valuation $u_v$, such that $a_v(t) = u_v$. The auction mechanism guarantees that their revenue $p^s_v(t)$ is greater than or equal to their true valuation $u_v$, i.e., $p^s_v(t) \geq u_v$. The pricing rules can be represented by the pricing vector of buyers $\mathbf{p}^b(t)$ and the pricing vector of sellers $\mathbf{p}^s(t)$.

\subsection{Problem Formulation}

{\color{black}To balance the demand and supply while reducing transmission delay }among buyers and sellers, i.e., balance the supply and demand, the mechanism needs to determine the allocation rules and pricing rules under the constraints of individual rationality (IR) and truthfulness. Individual rationality in the auction means that every buyer and seller interested in the auction benefits from their participation, i.e. each participant receives a non-negative utility from their involvement. In particular, when vehicle $v$ is a buyer, it pays a price $p_v^B$ that is the same as or lower than its value $u_v$, i.e., $p_v^B \leq u_v$. When vehicle $v$ is a seller, it receives a payment $p_v^S$ that is the same as or higher than its value $u_v$. 

Truthfulness in auctions refers to the property where buyers and sellers submit bids that reflect their true valuations. This ensures that the auction results are accurate and efficient, with all participants providing an honest assessment of the value of the goods or services being traded. For buyers, a truthful bid $b_v(t)$ corresponds to their true valuation $u_v$, meaning that $b_v(t) = u_v$. The auction mechanism ensures that their payment $p^b_v(t)$ is less than or equal to their true valuation $u_v$, i.e., $p^b_v(t) \leq u_v$. For sellers, a truthful ask $a_v(t)$ represents their true valuation $u_v$, such that $a_v(t) = u_v$. The auction mechanism guarantees that their revenue $p^s_v(t)$ is greater than or equal to their true valuation $u_v$, i.e., $p^s_v(t) \geq u_v$. In the double-sided market, truthfulness is achieved by requiring that buyers' values for the item being auctioned are greater than or equal to the price they will pay for it and that sellers' prices are greater than or equal to the value they place on the item being sold.

The global social welfare $SW(t)$ at time slot $t$, i.e., the sum of the values of allocated sellers and buyers, can be calculated as 
\begin{equation}
\begin{aligned}
    SW(t) =& \sum_{v\in \mathcal{V}^B(t)}\sum_{v'\in \mathcal{V}^S(t)} x_{v,v'}(t) u_v(t) \\&+ \sum_{v\in \mathcal{V}^B(t)} \sum_{v'\in \mathcal{V}^S(t)} y_{v,v'}(t) u_{v'}(t).
\end{aligned}
\end{equation}
Social welfare can be used to evaluate the efficiency of the mechanism, where social welfare is highest when the balance of supply and demand can be achieved. This indicates all the potential demands of buyers, i.e., requesting vehicles, and all the potential supplies of sellers, i.e., responding vehicles, can be matched and realized.
In addition, the total transmission latency $L(t)$ at time slot $t$ can be calculated as
\begin{equation}
\begin{aligned}
    L(t) = &\sum_{v \in \mathcal{V}^B(t)} \sum_{n \in \mathcal{N}} g_{n,v}(t) \left[1-\sum_{v' \in \mathcal{V}^S(t)} x_{v',v}(t)\right] \frac{c_v(t)}{R^d_{n,v}(t)}  \\&+ \sum_{v \in \mathcal{V}^B(t)}\sum_{v' \in \mathcal{V}^S(t)} x_{v',v}(t) \frac{c_v(t)}{R^c_{v',v}(t)},
\end{aligned}
\end{equation}
where the total transmission latency consists of the direct communication latency between RSUs and vehicles and the cooperative communication latency between vehicles and vehicles. {\color{black}The direct communication latency is calculated as the sum of the transmission latency of losing requesting vehicles that need to request the content directly from the RSU. Meanwhile, the cooperative communication latency is calculated as the sum transmission latency of winning requesting vehicles that received the content from relaying vehicles.}

To maximize the social welfare in the IoV data-sharing market and minimize the latency during data transmission, the optimization problem can be formulated as

\begin{maxi!}|s|[2]<b>
    {\mathcal{M}}{\frac{1}{T} \sum_{t\in \mathcal{T}} \big(SW(t) - L(t)\big)}{}{}
    \addConstraint{\sum_{n\in\mathcal{N}} g_{v,n}(t)}{=1}{\forall v\in \mathcal{V}, t \in \mathcal{T}\label{con1}}{}{} % 每个车辆只能加入一个RSU的本地市场
    \addConstraint{\sum_{v\in \mathcal{V}^B(t)} y_{v',v}(t)}{\leq 1}{\forall v'\in \mathcal{V}^S(t), t \in \mathcal{T} \label{con2}}{}{}
    \addConstraint{ x_{v',v}(t)}{\leq y_{v',v}(t)}{\forall v \in \mathcal{V}^B(t), v' \in \mathcal{V}^S(t), t \in \mathcal{T} \label{con3}}{}{}. % 买的东西少于卖的东西
\end{maxi!}

The constraint \eqref{con1} indicates that each vehicle can participate in one local market at any time. The constraint \eqref{con2} guarantees that the resources of RSUs and selling vehicles at any time can only be sold once. Finally, the constraint \eqref{con3} indicates the allocation variables of the mechanism.

\begin{figure}[t]
    \centering
    \includegraphics[width=0.8\linewidth]{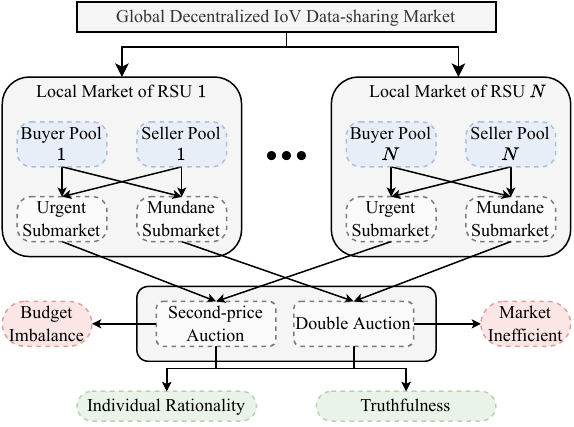}
    \caption{The workflow of the proposed hierarchical auction-based mechanism for the decentralized IoV data-sharing market, where buyers can enter either urgent submarket or mundane submarket by indicating their strategies.}
    \label{fig:mechanism}
    \vspace{-5mm}
\end{figure}
\subsection{Mechanism Design}
In this subsection, we provide a detailed description of the proposed hierarchical auction-based mechanism. Each RSU sets up a local market within its coverage area, which includes all vehicles within its coverage. {\color{black} Within each local market, there are two types of submarkets, i.e., the single-side urgent submarket and the double-side mundane submarket. In the urgent submarket, the transactions can be cleared as the buying bids and the selling bids are matched. In the mundane submarket, transactions are determined periodically. Therefore, when the traders are urgent for their content, they participate in the urgent market for instant transactions. When the traders are mundane, they participate in the mundane market and wait for the clearance of the market. In these local markets, we consider participants as passive sellers and active buyers. The passive sellers indicate that the relaying vehicles passively participate in the hierarchical submarket of the RSU. Active buyers can select the submarket in which they are willing to participate in.}  At each time slot $t$, buying vehicle $v$ submits its buying bid $d^b_v(t)$ to request content, and selling vehicle $v'$ submits its selling bid $d^s_{v'}(t)$ for the content response.

The sellers in the local submarket of RSU $n$ are organized into a seller pool $\mathcal{P}^S_{n}(t) = \{ v | \sum_{v\in \mathcal{V}^B(t)} y_{v',v}(t) = 0, g_{n,v'}(t)=1, v'\in \mathcal{V}^S(t)\}$ to passively respond to requests from buyers who arrive continuously in the market. The selling bids of sellers in the seller pool are in a seller value pool, which can be represented as $\mathcal{D}^S_{n}(t) = \{ d^s_{v'}(t) | \sum_{v\in \mathcal{V}^B(t)} y_{v',v}(t) = 0, g_{n,v'}(t)=1, v'\in \mathcal{V}^S(t)\}$. Without loss of generality, we assume that the selling bids in the seller value pool are sorted by their bids $d_{v}^s(t) \leq d_{v+1}^s(t),  \forall v, v+1 \in \mathcal{P}^S_{n}(t)$.
The buying vehicles that enter the mundane submarket are gathered into the buyer pool $\mathcal{P}^B_{n}(t)$ and their buying bids are gathered into the buyer value pool $\mathcal{D}^B_{n}(t)$. 
Without loss of generality, we assume that the selling bids in the seller value pool are sorted by their bids $d_{v}^b(t) \geq d_{v+1}^b(t),  \forall v, v+1 \in \mathcal{P}^B_{n}(t)$.
In the urgent submarket, the seller pool can promptly respond to incoming purchase requests from buying vehicles $\mathcal{P}^B_n(t)\backslash \mathcal{P}^B_{n}(t)$. Meanwhile, when buyers enter the mundane submarket, they can be organized into a buyer pool $\mathcal{P}^B_{n}$ and then periodically cleaned for the bilateral market. 
% Without loss of generality, we assume that the sellers in the seller pool are sorted by their bids $d_{v}^s(t) < d_{v+1}^s(t),  \forall v, v+1 \in \mathcal{P}^S_{n}(t)$, without loss of generality. On the other hand, buying vehicles may enter either the urgent or the mundane submarket.
\begin{mechanism}
In the mechanism, the urgent submarket on the buyer side uses a second-price auction, while the mundane submarket on both sides uses McAfee's mechanism~\cite{mcafee1992dominant}. Below are the detailed allocation and pricing rules:
\begin{enumerate}

\item Single-side Urgent Submarket: The auctioneer then determines the allocation of the task between sellers and the transaction price using the following steps. First, the auctioneer evaluates the seller pool and modifies a supply and demand matrix for vehicle $v'$. The seller with the highest bid is allocated, which can be calculated as
    \begin{equation}
    x_{v,v'}(t) = y_{v',v}(t) = 1_{\{d^s_{v'}(t)> \max\{\mathcal{D}^S_{-v'}(t)\}\}},
    \end{equation}
    where $\mathcal{D}^S_{-v'}(t)$ indicates the highest price in the seller pool without seller $v'$.
    Then, the auctioneer determines the transaction price for buying vehicle $v$ in the urgent market based on the second-price sealed-bid auction, which can be represented as
    \begin{equation}
        p_{v}^b(t) = \sum_{v' \in \mathcal{P}^S(t)} x_{v,v'}(t) \cdot \max\{\mathcal{D}^S_{n,-v'}(t)\},
    \end{equation}
    and the revenue received by the selling vehicle $v'$ can be represented as
    \begin{equation}
        p_{v'}^s(t) =y_{v',v}(t) \max\{\mathcal{D}^S_n(t)\},
    \end{equation}
    where $\mathcal{D}^S_{n,-v'}(t)$ indicates the selling bids in the selling pool without the bid of selling vehicle $v'$.
\item Double-side Mundane submarket: The double-side mundane submarket consists of the buyer pool $\mathcal{P}^B(t) = \{v|g_{v,n}=1, \sum_{v'\in \mathcal{V}^S(t)} x_{v,v'}(t)=0, v\in \mathcal{V}^B(t)\}$ and the untraded sellers in seller pool $\mathcal{P}^S(t)$ is cleared periodically. Based on McAfee's mechanism~\cite{mcafee1992dominant}, the auctioneer determines the allocation rule and the pricing rule as follows. For the buyer pool and seller pool at time slot $t$, the auctioneer sorts the buyer and sellers in the pools in the natural and finds the breakeven index $K$ in $\mathcal{D}^S_{n}(t)$ and $\mathcal{D}^B_{n}(t)$. Then, the auctioneer calculates the average price $p(t)=(b_{k+1}+s_{k+1})/2$. If $\sum_{v\in \mathcal{P}^B(t)}1_{\{b_v(t) \geq p(t)\}} = K$ and $\sum_{v'\in \mathcal{P}^S(t)}1_{\{a_v(t) \leq p(t)\}} = K$, then for the first $K$-th buyers and sellers in the buyer pool and the seller pool, the allocation rules are set to $x_{b_k, s_k}(t) = x_{b_k, s_k}(t) = 1$, for $k = 1,\ldots, K$. The pricing rule indicates that the clearing price is set to $p_{b_k}^b(t) = p_{s_k}^s(t) = p(t)$, for $k = 1,\ldots, K$. 
Otherwise, the first $K-1$-th sellers trade for $s_k$ and the first $k-1$ buyers trade for $b_k$ as in the trade-reduction mechanism. The allocation rules are set to $x_{b_k, s_k}(t) = x_{b_k, s_k}(t) = 1, k = 1,\ldots, K-1$. The pricing rule indicates that the clearing price is set to $p_{b_k}^b(t) = p_{K}^b(t)$ and $p_{K}^s(t) = p_{K}^s(t)$, for $k = 1,\ldots, K-1$. 
\end{enumerate}
\end{mechanism}

After all the local markets are clear, the local auctioneers of RSUs can measure the local budget cost for its local market. The local budget cost $\beta_n(t)$ of RSU $n$ can be calculated as 
\begin{equation}
\begin{aligned}
    \beta_n(t) = &
    \sum_{v\in\mathcal{V}^B(t)} g_{v,n}(t) x_{v,v'}(t) p_v^b(t) \\&- \sum_{v'\in\mathcal{V}^S(t)}g_{v',n}(t) y_{v,v'}(t) p_v^s(t)
\end{aligned}
    % \sum_{v\in\mathcal{V}^B(t)}\sum_{v'\in\mathcal{V}^S(t)} g_{v,n} g_{v',n} (x_{v,v'}(t) p_v^b(t) - y_{v',v}(t) u_v(t)).
\end{equation}
% After the market is cleared, the local auctioneers of RSUs can calculate the budget cost for its local market. To sum up the local budget costs of the decentralized market, the global budget cost can be calculated as
% \begin{equation}
%     B(t) = \frac{1}{T}\sum_{t\in \mathcal{T}}\sum_{n\in\mathcal{N}} \sum_{v\in\mathcal{V}} g_{v,n} (x_v(t) u_v(t) - y_v(t) u_v(t)).
% \end{equation}
{\color{black}Then, the total budget cost in the global decentralized market is the sum of the local budget cost $\beta_n(t), \forall n\in\mathcal{N}$, which can be calculated as}
\begin{equation}
    \beta(t) = \sum_{n\in\mathcal{N}} \beta_n(t).
\end{equation}

To minimize the total budget cost while maintaining social welfare in the decentralized market, we then let buying vehicles act as learning agents to learn submarket participating strategies for the equilibrium of supply and demand.

\subsection{Partially Observable Markov Decision Process}

In this system, each buyer is considered a learning agent in a Partially observable Markov decision process (POMDP), which can be characterized by the following components:

\subsubsection{Observation} The observation can be extracted from the state $S(t)$ of the system. The observation of vehicle $v$ at each time step $t$ includes the number of buyers and sellers in each local market, denoted by $|V^B_n(t)\bigcup V^S_n(t)|, \forall n \in \mathcal{N}$, respectively, the price of the last transaction is denoted by $\bar{p}(t-1)$, while the transmission rate with RSUs and relay vehicles can be represented by $R_{v}^{d}(t)$ and $R^{c}_{v,\bar{v}}(t)$, which can be represented as
\begin{equation}
\begin{aligned}
    O_v(t) = \{|V^B_1(t)\cup V^S_1(t)|, &\ldots, |V^B_N(t)\cup V^S_N(t)|, \\&\bar{p}(t-1), R_{v}^{d}(t), R^{c}_{v,\bar{v}}(t)\}
\end{aligned}
\end{equation}
\subsubsection{Action} The action of vehicle $v$ at time slot $t$ is represented by the strategy $A_v(t)=\{0,1\}$ {\color{black}in} which market the buyer enters. When $A_v(t)=0$, the vehicle $v$ participates in the urgent submarket. When $A_v(t)=1$, the vehicle $v$ enters the buyer pool and participates in the mundane submarket.  
\subsubsection{Reward} The reward function consists of social welfare, budget and transmission latency at the current time slot, which can be calculated as
\begin{equation}
    R_v(S(t), A_v(t)) = SW(t) - \alpha \beta(t)^2 - L(t),
\end{equation}
where $\alpha$ is the coefficient to scale the budget cost.
Specifically, social welfare is represented by the total utility of all buyers and sellers in the system, which is a function of the prices paid and received for each transaction. The budget is represented by the difference between the total payment received from all winning bidders and the total price paid to all winning sellers. Transmission latency is represented by the delay caused by the transmission of content over the IoV.
\subsubsection{Value Function} Given policy $\pi_v$ of vehicle $v$, value function $V_{\pi_v}(S(t))$ of state $S(t)$, the expected return when starting in $S$ and following $\pi_v$, can be formulated by
	\begin{equation}
	V_{\pi_v}(S) := \mathbb{E}_\pi\left[\sum_{t=0}^{T}\gamma^k R_v(S(t), A_v(t))|S^0=S\right],
	\end{equation}
	where $\mathbb{E}_\pi(\cdot)$ denotes the expected value of a random variable given that the {\color{black}learning agent} follows policy $\pi$ and $\gamma\in[0,1]$ is the discount factor for rewards used to reduce weights as the time step increases.

The POMDP framework provides a mathematical model that can be used to optimize the buyer's decision-making process. By observing the outcome of the system and taking appropriate action, the agent can maximize its expected reward over the long term. This framework can be used to design learning agents that can participate in the P2P data trading system while maximizing their own utility and contributing to the overall welfare of the system. To maximize the value function, Multi-agent Proximal Policy Optimization (PPO), a type of reinforcement learning algorithm, is used to train multiple agents to perform tasks in environments where feedback is provided in the form of a reward signal. In this context, each agent is trained using the POMDP framework. The training process involves observing the current state of the system and taking appropriate action to maximize the expected reward over the long term. Let $\theta_v$ be the parameters in the policy network of vehicle $v$ and $\phi_v$ be the parameters in the value network. In MAPPO, each agent maintains its policy $\pi_{v}(\theta_v)$, which is updated using a centralized critic $V(s;\phi_v)$ that estimates the value of the global state $S$. The critic and policy network is trained using the clip loss $L^\emph{CLIP} (\theta_v, \phi_v)= L^\emph{P}(\theta_v) +  L^\emph{V}(\phi_v)$, which consists of the loss of policy networks and value networks~\cite{schulman2017proximal} in the global state.

{\color{black}Following the properties of the second-price auction and Macfee's double auction~\cite{niyato2020auction}, we demonstrate that the proposed decentralized hierarchical auction is IR and truthfulness in a decentralized market. Firstly, we need to prove that the proposed mechanism is IR and truthful in each local market. 

\begin{lemma}
Under the determined market entry strategies $A_v, \forall v\in \mathcal{V}^B(t)$, the mechanism is IR and truthfulness in each local market.
\end{lemma}}

{\color{black}This lemma can be obtained straightforwardly based on the properties of the second-price auction and the dominant strategy double auction~\cite{mcafee1992dominant}.} Then, as the mechanism can maintain IR and truthfulness in local markets, we demonstrate that the properties in local markets can also be extended to the global decentralized market.

\begin{theorem}
In a global decentralized market consisting of multiple local markets, the proposed mechanism is individually rational and truthful.
\end{theorem}
\begin{proof}
{\color{black}To prove individual rationality and truthfulness, we need to show the following properties hold for both the urgent and the mundane submarkets. As the market entry strategies are determined, the buyer $v$ with $A_v(t) = 0$ entering urgent submarkets always enters the urgent submarket in each local market, and with $A_v(t) = 1$ entering mundane submarkets enter the mundane submarket in each local market.}

{\color{black}First, we demonstrate that the proposed auction mechanism is IR for all participants. For the second-price auction in the urgent submarkets, the highest bidder pays the second-highest bid, which is less than or equal to their valuation. Thus, the winning buyers' utilities are non-negative, i.e., 
\begin{equation}
    \mu_v(t)(b_v(t)) = u_v(t) - p_v^b(t) \geq 0, \text{for} \sum_{v'\in\mathcal{V}^S(t)} x_{v,v'}=1
\end{equation}
for $\forall v\in \mathcal{V}^B(t).$ For the losers, their utility is zero since they do not need to pay, i.e.,
\begin{equation}
    \mu_v(t)(b_v(t)) = u_v(t) - p_v^b(t) = 0, \text{for} \sum_{v'\in\mathcal{V}^S(t)} x_{v,v'}=0
\end{equation}
for $\forall v\in \mathcal{V}^B(t).$ Therefore, the second-price auction is individually rational. For the double auction in the mundane submarkets, buyers pay a price less than or equal to their bid, and sellers receive a price greater than or equal to their ask. Since bidders only participate when their valuations are met or exceeded, their utility is non-negative. Thus, the double auction is IR.}

{\color{black}For the second-price auction in the urgent submarkets, truth-telling is a dominant strategy. If a bidder submits a bid $d_v^b(t)'$ higher than their true valuation $u_v(t)$, they risk winning the auction and paying more than their valuation, which results in negative utility $\mu_v(t)(d_v^b(t)') = u_v(t) - p_v^b(t) < 0$. If they submit a bid lower than their true valuation, they risk losing the auction even if they could have won, i.e., $\mu_v(t)(d_v^b(t)') = 0$, which may result in a lower utility. For the double auction in the mundane submarkets, buyers submitting a bid $d_v^b(t)'$ higher than their true valuation risk paying more than their valuation, i.e., $\mu_v(t)(d_v^b(t)') = \mu_v(t)(b_v(t)), \forall v\in\mathcal{V}^B(t)$ whereas submitting a bid lower than their true valuation risks missing out on a potentially profitable trade, i.e., $\mu_v(t)(d_v^b(t)')=0, \forall v\in\mathcal{V}^B(t)$. Similarly, sellers submitting an ask $d_v^s(t)'$ lower than their true valuation receive the same revenue as their valuation, i.e., $\mu_v(t)(d_v^s(t)') = \mu_v(t)(d_v^s(t)), \forall v\in\mathcal{V}^S(t)$, and submitting an ask $\bar{d}_v^s(t)$ higher than their true valuation risk missing out on a potentially profitable trade, i.e, $\mu_v(t)(\bar{d}_v^s(t))=0$. Therefore, truth-telling is the best strategy in second-price auctions and double auctions. Since both submarkets are truthful, the hierarchical auction mechanism is truthful.}
\end{proof}

Besides the IR and truthfulness, in the next section, we experimentally show that the proposed mechanism can improve market efficiency and reduce budget imbalance while reducing transmission latency in the IoV data-sharing market.

\section{Experimental Results}\label{sec: experiment}

In our experiments, we consider a {\color{black}1km$\times$1km} simulated vehicular network environment consisting of 4 RSUs, each of which covers an area of 500 m and a set of vehicles traveling among these RSUs with an average speed of 90 Km per hour. {\color{black} The simulator is established specifically to meet the evaluation methodology for the urban case outlined in Annex A of 3GPP TR 36.885~\cite{liang2019spectrum}.} The requests of vehicles are sampled from $[0, 1]$. When the required quality of vehicle $v$ is $1$, the vehicle is the buyer in the market otherwise it is the seller. The number of chucks required by each vehicle $v$ is randomly sampled from $[1, 10]$, and its value $u_v$ is calculated as $\log(1+c_v/10)$. The transmit power $p_v$ of selling vehicle $v$ is sampled from $U[0,10]$ mW and its value $u_v$ in proportion to the transmit power. The learning rate of DRL agents is set to $0.001$ and the discounting factor is set to $0.95$. We test the performance of the proposed MADRL-based mechanism in the systems with $20, 30, 40, 50, 60, 70,$ and $80$ vehicles. The coefficient of the value function is set to $0.5$, the entropy coefficient is set to $0.02$, and the clip factor is set to $0.2$. 

\subsection{Convergence Analysis}
\begin{figure}[t]
    \centering
    \includegraphics[width=0.8\linewidth]{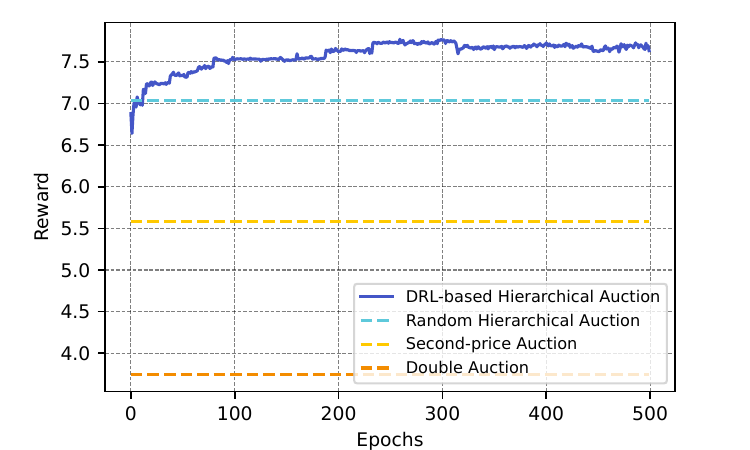}
    \caption{Training reward {\color{black}versus} training epochs.}
    \label{fig:convergence}
    \vspace{-5mm}
\end{figure}

\begin{figure}[t]
    \centering
    \includegraphics[width=0.8\linewidth]{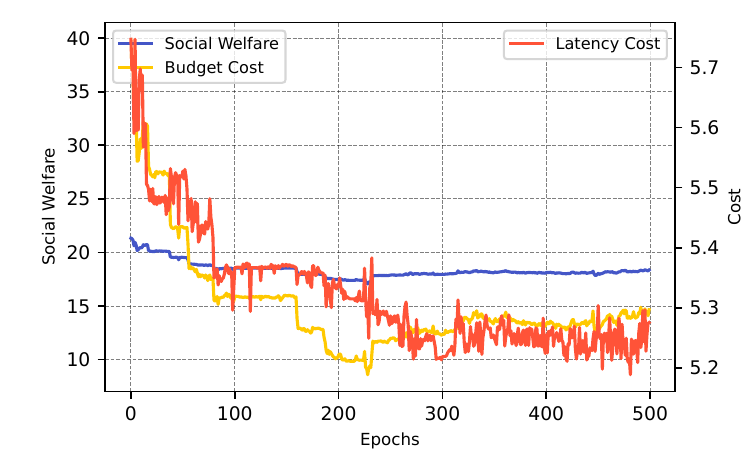}
    \caption{Detailed cost structure (social welfare, budget cost, and transmission delay) {\color{black}versus} training epoch.}
    \label{fig:convergence_eq}
    \vspace{-5mm}
\end{figure}
We first conduct experiments to {\color{black}show that} the MADRL-based mechanism can converge to satisfactory performance in the decentralized continuous data-sharing market. The results in Fig. \ref{fig:convergence} demonstrate that the algorithm converges to a satisfactory reward after approximately 300 epochs. The proposed algorithm outperforms random algorithms after only 10 epochs of training, demonstrating its efficiency and effectiveness. One of the major advantages of using DRL training is that rewards in the market can be significantly improved by up to 10\% compared to random strategies. The reason is that the proposed MADRL algorithm can learn and optimize its strategies based on past experiences and market conditions. More detailed analysis is provided in Fig.~\ref{fig:convergence_eq} for the performance of the proposed algorithm, social welfare, and budget costs in the market decreasing continuously during training. This is because there are fewer two-price auctions in the emergent market, leading to a decrease in social welfare and budget cost. Additionally, the latency cost in the market also decreases with the increase of training epochs, although it has no direct impact on market efficiency. Overall, the proposed MADRL-based mechanism offers significant benefits to vehicular network {\color{black}providers}. Not only can it reduce the latency cost of content transmission, but it can also maintain market efficiency, making it a valuable tool for market selection and optimization.
% First, we perform a convergence analysis for the proposed DRL-based continuous double auction. As shown in Fig.~\ref{fig:convergence}, the proposed DRL algorithm converges to a lower cost as the training epoch progresses. At the beginning of the training process, the cost of the algorithm is comparable to that of the random method, and at about 400 epochs, the proposed algorithm converges to a cost of about 2. Compared to the greedy algorithm, the proposed algorithm can reduce the cost by about 80\%, while compared to the random method, it can reduce the cost by about 75\%.

\subsection{Performance Comparison}
\begin{figure}[t]
    \centering
    \includegraphics[width=0.8\linewidth]{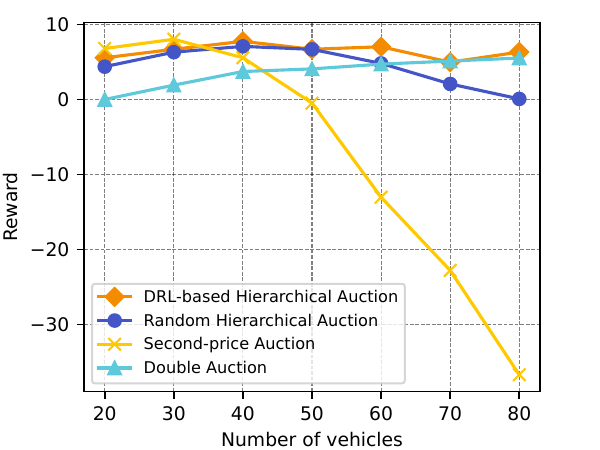}
    \caption{Reward {\color{black}versus} number of vehicles.}
    \label{fig:reward}
    \vspace{-5mm}
\end{figure}

\begin{figure}[t]
    \centering
    \includegraphics[width=0.8\linewidth]{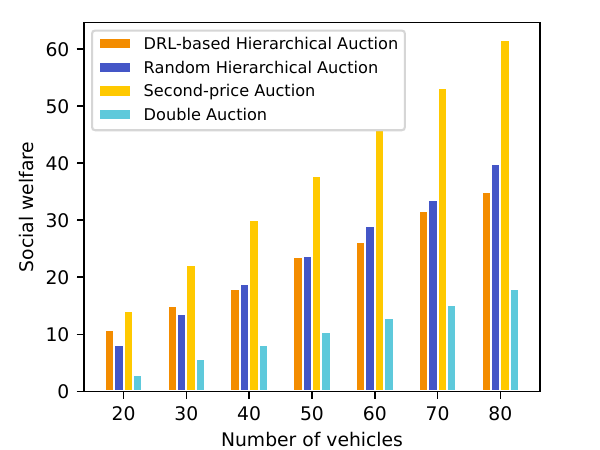}
    \caption{Social welfare {\color{black}versus} number of vehicles.}
    \label{fig:social}
    \vspace{-5mm}
\end{figure}

\begin{figure}[t]
    \centering
    \includegraphics[width=0.8\linewidth]{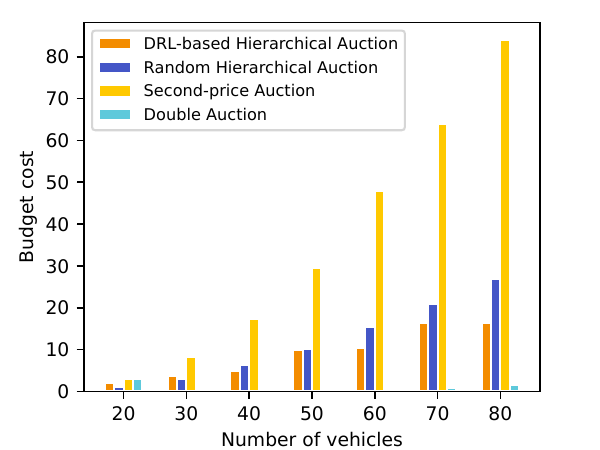}
    \caption{Budget cost {\color{black}versus} number of vehicles.}
    \label{fig:budget}
    \vspace{-5mm}
\end{figure}

\begin{figure}[t]
    \centering
    \includegraphics[width=0.8\linewidth]{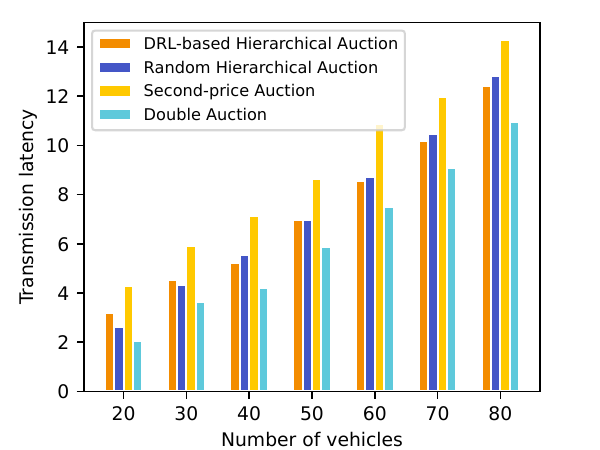}
    \caption{Transmission latency {\color{black}(ms)} versus number of vehicles.}
    \label{fig:latency}
    \vspace{-5mm}
\end{figure}
The rewards achieved by the proposed mechanism and other baselines are shown in Fig. \ref{fig:reward}. The performance of the proposed MADRL-based mechanism remains stable as the number of vehicles increases. When the number of vehicles is small, the second-price auction can yield a high reward, but its performance deteriorates rapidly as the number of vehicles grows. The random mechanism exhibits a similar trend to the second-price auction, as all of these mechanisms cannot fully utilize the information in different local markets. On the other hand, the performance of the double auction is {\color{black} poor} when the number of vehicles is small, {\color{black} as it can merely achieve limited social welfare, which is the sum of utilities of all entities in the system}. However, as the local market size grows, the performance of the double auction improves, as the price information is sufficient to balance the supply and demand. Nonetheless, the proposed MADRL-based mechanism can always achieve stable and satisfactory performance across different market sizes.

\textcolor{black}{Social welfare refers to the sum of utilities of all entities within a system, and it holds significant importance as it captures the collective utilities of all stakeholders engaged, rather than focusing solely on individuals.} From Fig. \ref{fig:social}, we can observe that the second-price auction consistently achieves the highest social welfare among all the mechanisms considered in our study. This is because the second-price auction encourages truthful bidding, leading to an efficient allocation of resources. On the other hand, the social welfare of the double auction varies based on the number of vehicles, which can be attributed to the fact that the double auction is more complex and involves more interactions among the bidders.
\textcolor{black}{In addition, we can observe that the double auction falls notably short in generating social welfare in comparison to alternative schemes. This observation indicates that, when provided with the same amount of resources, double auction produces lower overall utilities, thereby translating into inefficiencies in resource utilization.} 
Interestingly, we find that the DRL-based mechanism and the random mechanism achieve similar social welfare, with their difference increasing as the number of vehicles increases. This suggests that the DRL-based mechanism may not be the best choice in settings where the computational cost is high.

Moving on to Fig. \ref{fig:budget}, we observe that the increase in the number of vehicles has a significant impact on the budget cost of the second-price auction. This is because the second-price auction involves more operations, which in turn incurs higher communication and computation costs. However, the budget cost of the double auction is almost unaffected and always zero, as the double auction does not require any communication or computation. The random mechanism has a smaller budget cost than the DRL-based mechanism, with the difference increasing as the number of vehicles increases. This is because the random mechanism is the simplest mechanism among all the mechanisms considered in our study, and it does not require any computation or communication.

Finally, in Fig. \ref{fig:latency}, we compare the transmission delay of each mechanism. The double auction incurs the least transmission delay as it facilitates fewer transactions, regardless of the number of vehicles. However, it is important to note that the double auction may not achieve the highest social welfare although it can achieve the lowest budget cost.
Compared with the random mechanism, our proposed DRL-based mechanism can effectively reduce transmission latency, which is achieved by leveraging the power of deep reinforcement learning to learn an optimal bidding function that minimizes the transmission delay while ensuring a desirable level of social welfare.

% \begin{figure}[t]
% \vspace{-0.5cm}
% \centering
% \subfigure[Social welfare versus number of vehicles.]{\includegraphics[width=1\linewidth]{socialwelfare.pdf}%
% \label{social}}
% % \hfil
% \subfigure[Budget cost versus number of vehicles.]{\includegraphics[width=1\linewidth]{budget.pdf}%
% \label{budget}}
% \subfigure[transmission delay versus number of vehicles.]{\includegraphics[width=1\linewidth]{latency.pdf}%
% \label{VMAF_BS}}
% \caption{Performance comparison under different market sizes.}
% \label{fig:compare}
% \end{figure}

\section{Conclusion \label{sec: conclusion}}
Secure and efficient data sharing in IoV systems is an essential component that contributes to the proliferation of IoV ecosystems. In this paper, we enhanced the trading efficiency and minimized transmission latency for efficient data sharing in the IoV by designing a decentralized market with continuous double auctions. {\color{black}The proposed incentive mechanism to encourage user participation in the market is based on MADRL, for which theoretical analysis and experimental results show that it outperforms both second-price auctions and double auctions by at least 10\% while reducing transmission latency by 20\%.} In addition, we propose a time-sensitive KP-ABE encryption mechanism to couple with NDN to protect data in IoV, which adds an additional layer of security to our proposed solution. The proposed KP-ABE scheme enhances the security of data sharing by limiting access to data based on validity time and reducing the risk of unauthorized access, while NDN improves resource utilization in the network.

\bibliographystyle{bibstyle}
\bibliography{references}

\section*{Acknowledgement}
This research is supported by the National Research Foundation, Singapore under its Strategic Capability Research Centres Funding Initiative. Any opinions, findings, and conclusions or recommendations expressed in this material are those of the author(s) and do not reflect the views of National Research Foundation, Singapore; Jiani Fan's research is partly supported by Alibaba Group through Alibaba Innovative Research (AIR) Program and Alibaba-NTU Singapore Joint Research Institute (JRI), Nanyang Technological University, Singapore.

\vspace{-1cm}

\section*{Biography}
\vskip -1\baselineskip plus -1 fil
\begin{IEEEbiography}[{\includegraphics[width=1in,height=1.25in,clip,keepaspectratio]{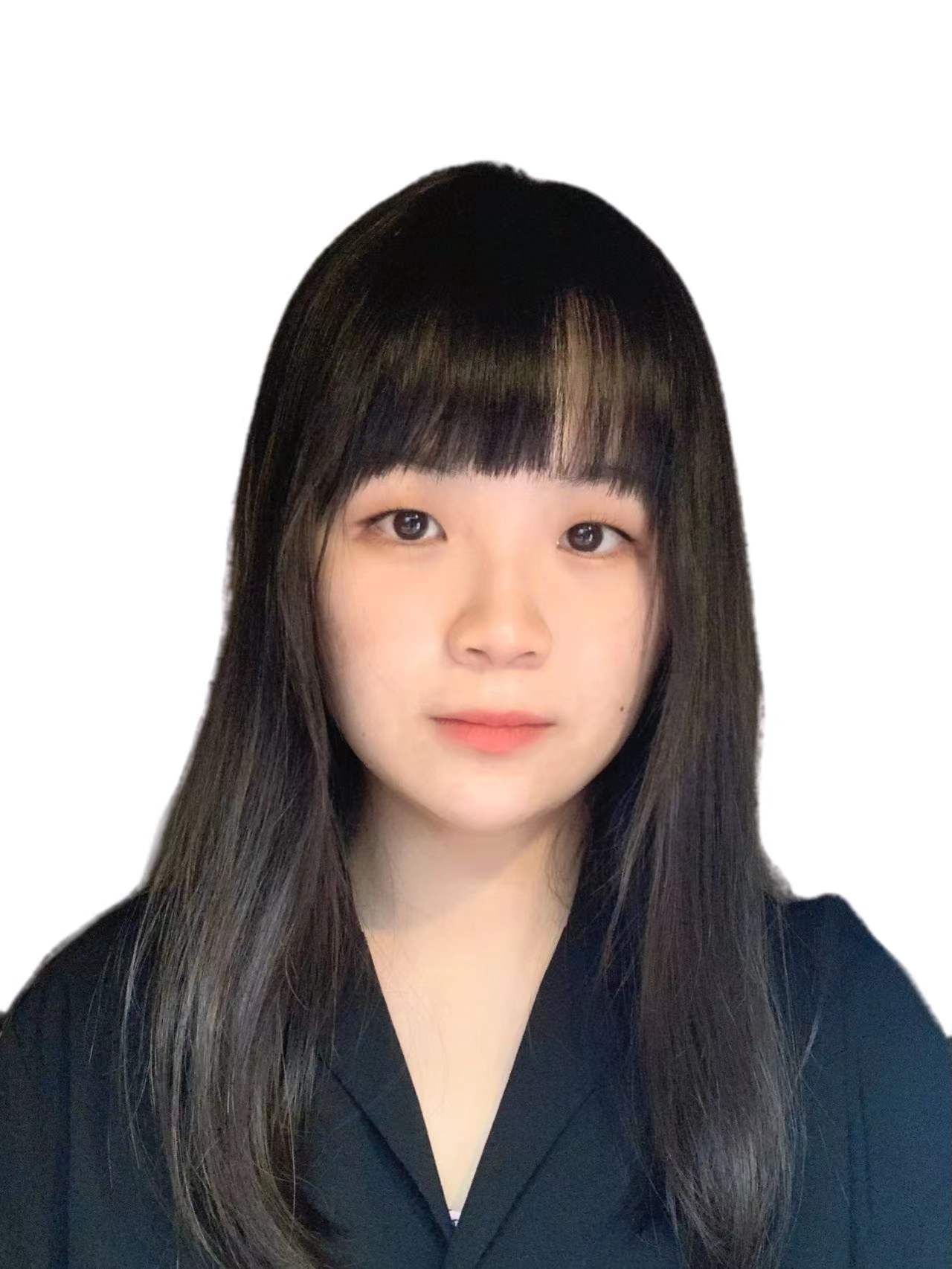}}]{Jiani Fan} received her B.S. from the School of Computing and Information Systems at Singapore Management University, Singapore. She is currently pursuing a Ph.D. degree at the School of Computer Science and Engineering, Nanyang Technological University, Singapore. Her research interests include IoT security, cybersecurity, and the Internet of Vehicles, with a focus on AI-powered security solutions.
\end{IEEEbiography}
\vspace{-0.5cm}

\vskip 0pt plus -1 fil
\begin{IEEEbiography}[{\includegraphics[width=1in,height=1.25in,clip]{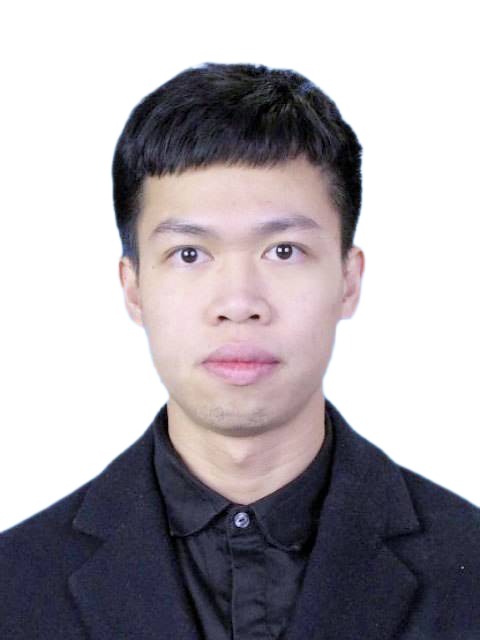}}]{Minrui Xu}
received the B.S. degree from Sun Yat-Sen University, Guangzhou, China, in 2021. He is currently working toward the Ph.D. degree in the School of Computer Science and Engineering, Nanyang Technological University, Singapore. His research interests mainly focus on Metaverse, deep reinforcement learning, and mechanism design.
 \end{IEEEbiography}
\vspace{-0.5cm}

\vskip 0pt plus -1 fil
\begin{IEEEbiography} [{\includegraphics[width=1in,height=1.25in,clip,keepaspectratio]{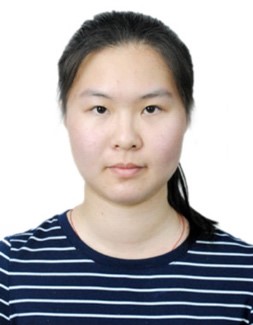}}]{Jiale Guo}
received the B.Sc. from Shandong University, China, in 2017, and the Ph.D. degree in computer science from Nanyang Technological University, Singapore, in 2022. She is currently a Research Fellow with the Strategic Centre for Research in Privacy-Preserving Technologies and Systems (SCRiPTS), Nanyang Technological University. Her research interests include privacy-preserving machine learning and cyber security. 
\end{IEEEbiography}
\vspace{-0.5cm}

\vskip 0pt plus -1 fil
\begin{IEEEbiography} [{\includegraphics[width=1in,height=1.25in,clip,keepaspectratio]{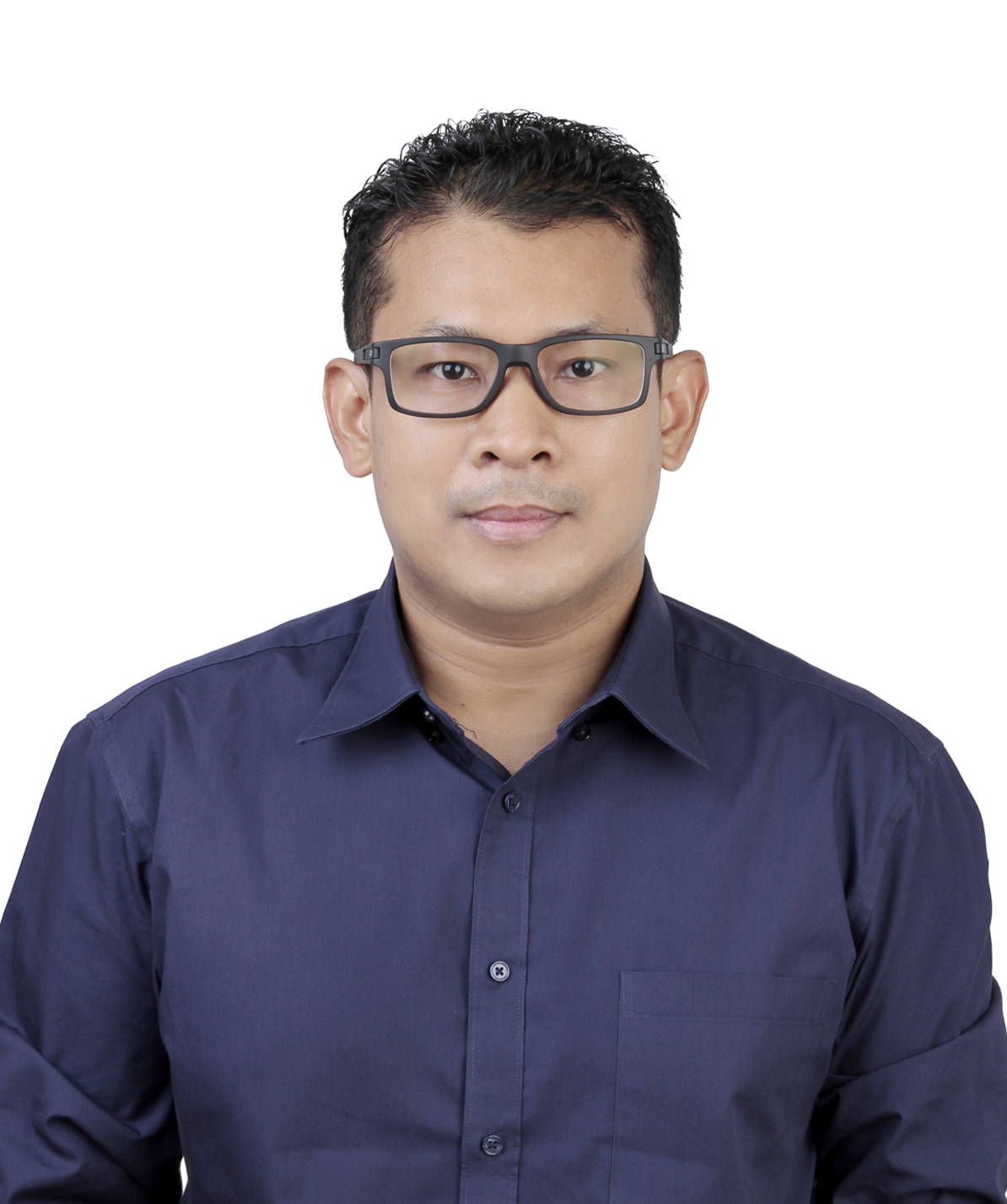}}]{Lwin Khin Shar}
is an Associate Professor in the School of Computing and Information Systems at Singapore Management University, Singapore. He received his Ph.D. degree in Software Engineering from Nanyang Technological University (NTU), Singapore in 2014.  He was a postdoctoral research associate at SnT of the University of Luxembourg and then a research scientist at NTU. His research interests span software engineering, security \& privacy, and machine learning, while specializing in analysis of web \& mobile applications and recently cyber-physical systems for detecting security vulnerabilities, privacy issues, malware, and anomalies. He is author or coauthor of more than 40 research papers published in international journals and conferences/workshops, including top venues (e.g., IEEE-TSE, IEEE-TDSC, EMSE, ICSE, FSE, ASE).  In his research, he often collaborates with industry partners spanning from healthcare and traffic management to Government sectors.
\end{IEEEbiography}
\vskip 0pt plus -1 fil
\vspace{-0.5cm}

\begin{IEEEbiography} [{\includegraphics[width=1in,height=1.25in,clip,keepaspectratio]{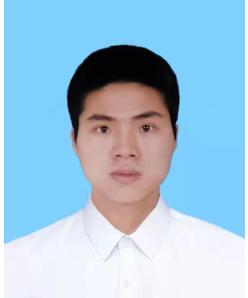}}]{Jiawen Kang} received the Ph.D. degree from the Guangdong University of Technology, China in 2018. He was a postdoc at Nanyang Technological University, Singapore from 2018 to 2021. He currently is a professor at Guangdong University of Technology, China. His research interests mainly focus on blockchain, security, and privacy protection in wireless communications and networking.
\end{IEEEbiography}
\vspace{-0.5cm}

\vskip 0pt plus -1 fil
\begin{IEEEbiography} [{\includegraphics[width=1in,height=1.25in,clip,keepaspectratio]{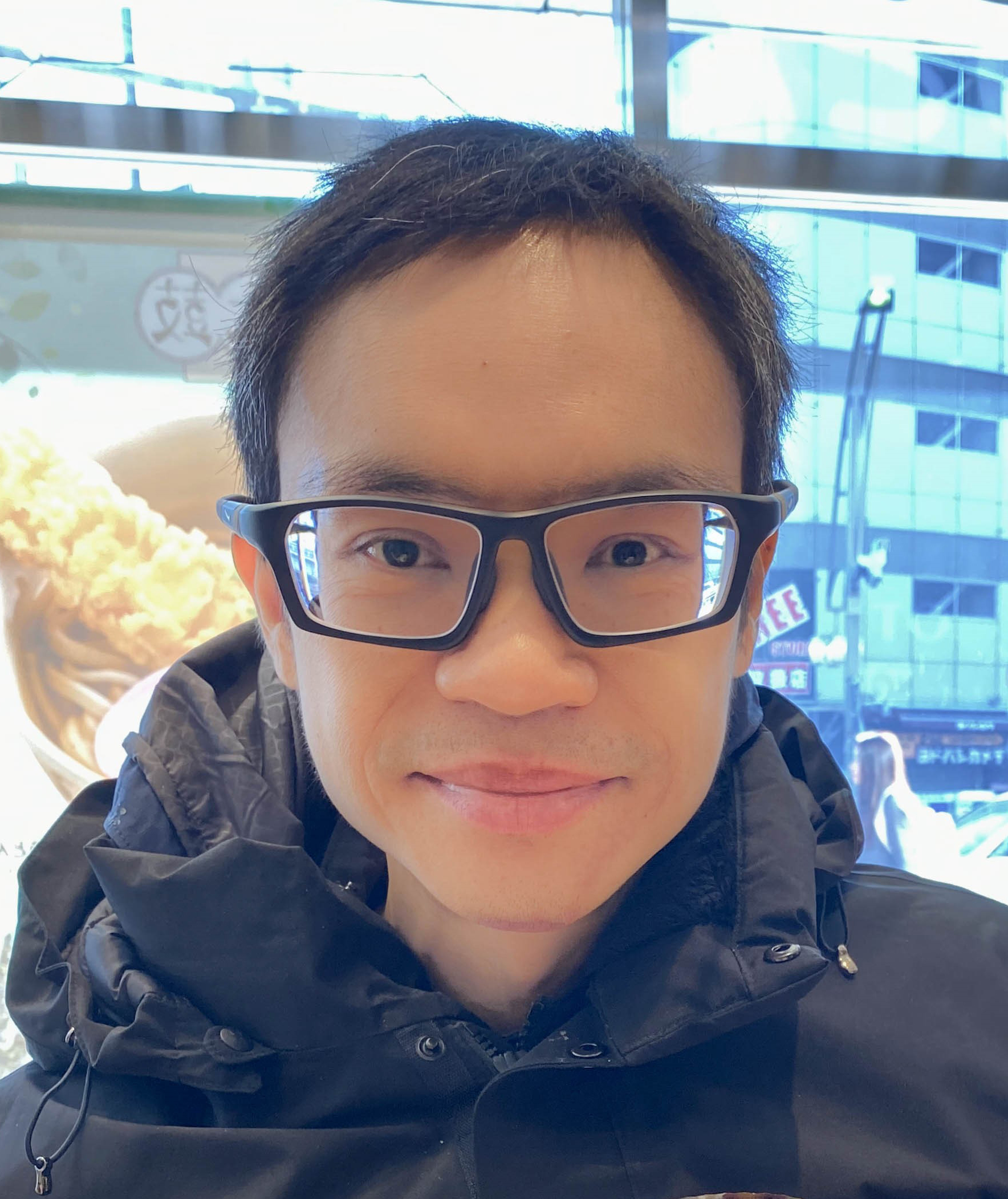}}]{Dusit Niyato} (Fellow, IEEE) is currently a professor in the School of Computer Science and Engineering, at Nanyang Technological University, Singapore. He received B.Eng. from King Mongkut's Institute of Technology Ladkrabang (KMITL), Thailand in 1999 and the Ph.D. in Electrical and Computer Engineering from the University of Manitoba, Canada in 2008. His research interests are in the areas of the Internet of Things (IoT), machine learning, and incentive mechanism design.
\end{IEEEbiography}
\vspace{-0.5cm}

\vskip 0pt plus -1 fil
\begin{IEEEbiography} [{\includegraphics[width=1in,height=1.25in,clip,keepaspectratio]{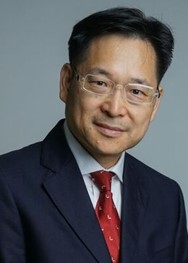}}]{Kwok-Yan Lam}
(Senior Member, IEEE) received his B.Sc. degree (1st Class Hons.) from University of London, in 1987, and Ph.D. degree from University of Cambridge, in 1990. He is the Associate Vice President (Strategy and Partnerships) and Professor in the School of Computer Science and Engineering at the Nanyang Technological University, Singapore. He is currently also the Director of the Strategic Centre for Research in Privacy-Preserving Technologies and Systems (SCRiPTS). From August 2020, he is on part-time secondment to the INTERPOL as a Consultant at Cyber and New Technology Innovation. Prior to joining NTU, he has been a Professor of the Tsinghua University, PR China (2002–2010) and a faculty member of the National University of Singapore and the University of London since 1990. He was a Visiting Scientist at the Isaac Newton Institute, Cambridge University, and a Visiting Professor at the European Institute for Systems Security. In 1998, he received the Singapore Foundation Award from the Japanese Chamber of Commerce and Industry in recognition of his research and development achievement in information security in Singapore. His research interests include Distributed Systems, Intelligent Systems, IoT Security, Distributed Protocols for Blockchain, Homeland Security and Cybersecurity.
\end{IEEEbiography}

\end{document}